\documentclass[pra,showpacs,twocolumn]{revtex4-1}
\usepackage{amssymb,amsmath,bm} 
\usepackage{mathrsfs}
\usepackage{graphicx,subfigure}

\makeatletter

\def \v#1{{\bm #1}}
\def \be {\begin{equation}}
\def \ee {\end{equation}}

\newcommand{\I}{\mathrm{i}}

\newcommand{\ket}[1]{|#1\rangle}
\newcommand{\bra}[1]{\langle#1|}
\newcommand{\tr}[1]{\mathrm{tr}\left(#1\right)}
\newcommand{\Tr}[1]{\mathrm{Tr}\left\{#1\right\}}
\newcommand{\Trabs}[1]{\mathrm{TrAbs}\left\{#1\right\}}

\def \ds{\displaystyle}
\def \trAbs{\mathrm{TrAbs}}
\def \Re{\mathrm{Re}\,}
\def \Im{\mathrm{Im}\,}

\newcommand{\vin}[2]{\langle#1,#2\rangle}
\newcommand{\vecin}[2]{\langle#1|#2\rangle}
\newcommand{\sldin}[2]{\langle#1,#2\rangle_{\rho_{\theta}}}
\newcommand{\rldin}[2]{\langle#1,#2\rangle_{\rho_{\theta}}^{+}}
\def \del{\partial}
\def \stheta{\v{s}_{\theta}}
\def \Def{\stackrel{\rm def}{\Leftrightarrow}}

\def \cB{{\cal B}}

\def \cM{{\cal M}}
\def \cH{{\cal H}}
\def \cX{{\cal X}}
\def \sofh{{\cal S}({\cal H})}

\def \bbr{{\mathbb R}}
\def \bbc{{\mathbb C}}

\def \sofc2{{\cal S}({\mathbb C}^2)}

\def \lofh{{\cal L}({\cal H})}
\def \lofhh{{\cal L}_h({\cal H})}
\def \cD#1{{\cal D}_{#1}}
\def \cW{{\cal W}}

\def \rhon{\rho_\theta^{\otimes n}}
\def \Eof#1{{\mathrm{E}}_{\theta}^{(n)}[#1]}

\newtheorem{theorem}{Theorem}[section]
\newtheorem{lemma}[theorem]{Lemma}
\newtheorem{proposition}[theorem]{Proposition}

\newenvironment{proof}[1][Proof:]{\begin{trivlist}
\item[\hskip \labelsep {\bfseries #1}]}{\end{trivlist}}

\newenvironment{definition}[1][Definition]{\begin{trivlist}
\item[\hskip \labelsep {\bfseries #1}]}{\end{trivlist}}

\newcommand{\qed}{\nobreak \ifvmode \relax \else
      \ifdim\lastskip<1.5em \hskip-\lastskip
      \hskip1.5em plus0em minus0.5em \fi \nobreak
      \vrule height0.75em width0.5em depth0.25em\fi}

\begin{document}

\title{Explicit formula for the Holevo bound for two-parameter qubit-state estimation problem}
\author{Jun Suzuki}
\date{\today}
\email{junsuzuki@uec.ac.jp}
\affiliation{Graduate School of Information Systems, The University of Electro-Communications,\\
1-5-1 Chofugaoka, Chofu-shi, Tokyo, 182-8585 Japan}

\begin{abstract}
The main contribution of this paper is to derive an explicit expression 
for the fundamental precision bound, the Holevo bound, 
for estimating any two-parameter family of qubit mixed-states in terms of quantum versions of Fisher information. 
The obtained formula depends solely on the symmetric logarithmic derivative (SLD), 
the right logarithmic derivative (RLD) Fisher information, and a given weight matrix. 
This result immediately provides necessary and sufficient conditions for the following 
two important classes of quantum statistical models; 
the Holevo bound coincides with the SLD Cram\'er-Rao bound 
and it does with the RLD Cram\'er-Rao bound. 
One of the important results of this paper is that a general model other 
than these two special cases exhibits an unexpected property: 
The structure of the Holevo bound changes smoothly when the weight matrix varies. 
In particular, it always coincides with the RLD Cram\'er-Rao bound for a certain choice of the weight matrix. 
Several examples illustrate these findings.  
\end{abstract}


\maketitle 

\section{Introduction}
One of the fundamental questions in quantum statistical inference problem is 
to establish the ultimate precision bound for a given quantum statistical model allowed by 
the laws of statistics and quantum theory. 
Mainly due to the non-commutativity of operators and nontrivial optimization over 
all possible measurements, 
this question still remains open in full generality. 
This is very much in contrast to the classical case where the precision bounds 
are obtained in terms of information quantities for various statistical inference problems. 

The problem of point estimation of quantum parametric models is 
of fundamental importance among various quantum statistical inference problems. 
This problem was initiated by Helstrom in the 1960s 
and he devised a method to translate the well-known strategies developed in classical statistics 
into the quantum case \cite{helstrom}. A quantum version of Fisher information 
was successfully introduced and the corresponding precision bound, a quantum version 
of Cram\'er-Rao (CR) bound, was derived.  
It turned out, however, that the obtained bound is not generally achievable except for trivial cases. 

A clear distinction regarding the quantum parameter estimation problem 
arises when exploring possible estimation strategies 
since there is no measurement degrees of freedom in the classical estimation problem. 
Consider $n$ identical copies of a given quantum state and 
we are allowed to perform any kinds of quantum measurements according to quantum theory. 
A natural question is then to ask how much can one improve estimation errors 
by measurements jointly performed on the $n$ copies when compared to the case 
by those individually performed on each quantum state. The former class of 
measurements is called {\it collective} or {\it joint} and the latter is referred to 
as {\it separable} in literature. It is clear that the class of collective measurements 
includes separable ones and one expects that collective measurements should be 
more powerful than separable ones in general.  
Since one cannot do better than the best collective measurement, 
the ultimate precision bound is the one that is asymptotically achieved 
by a sequence of the best collective measurements as the number of copies tends to infinity. 
This fundamental question has been addressed by several authors before 
\cite{nagaoka89-2,HM98,GM00,BNG00,hayashi03,hayashi,bbmr04,bbgmm06,HM08,GK06,KG09,YFG13,GG13}. 

It was Holevo who developed parameter estimation theory of quantum states 
by departing from a direct analogy to classical statistics. 
He proposed a bound, known as the Holevo bound, in the 1970s aiming to 
derive the fundamental precision limit for quantum parameter estimation problem, see Ch.~6 of his book \cite{holevo}. 
At that time, it was not entirely clear whether or not this bound is a really 
tight one, i.e., the asymptotic achievability by some sequence of measurements. 
Over the last decade, there have been several important progress on asymptotic analysis 
of quantum parameter estimation theory 
revealing that the Holevo bound is indeed the best asymptotically achievable bound 
under certain conditions \cite{HM08,GK06,KG09,YFG13}. 
These results confirm that the Holevo bound plays 
a pivotal role in the asymptotic theory of quantum parameter estimation problem. 
Despite the fact that we now have the fundamental precision bound, 
the Holevo bound has a major drawback: It is not an explicit form 
in terms of a given model, but rather it is written as an optimization of 
a certain nontrivial function. 
Therefore, unlike the classical case, where the Fisher information 
can be directly calculated from a given statistical model, 
the structure of this bound is not transparent in terms of the model under consideration. 

Having said the above introductory remarks, we wish to gain a deeper insight into 
the structure of the Holevo bound reflecting statistical properties of a given model. 
To make progress along this line of thoughts, we take the simplest quantum parametric model, 
a general qubit model, and analyze its Holevo bound in detail. 
Since explicit formulas for the Holevo bound for mixed-state models with one and three parameters 
and pure-state models are known in literature, 
the case of two-parameter qubit model is the only one left to be solved. 
The main contribution of this paper is to derive an explicit expression 
for the Holevo bound for any two-parameter qubit model of mixed-states without 
referring to a specific parametrization of the model. 
Remarkably, the obtained formula depends solely on a given weight matrix and three previously known bounds: 
The symmetric logarithmic derivative (SLD) CR bound, 
the right logarithmic derivative (RLD) CR bound, and the bound for D-invariant models. 
This result immediately provides necessary and sufficient conditions for the two important cases. 
One is when the Holevo bound coincides with the RLD CR bound 
and the other is when it does with the SLD CR bound. 
We also show that a general model other than these two special cases exhibits an unexpected 
property, that is, the structure of the Holevo bound changes smoothly when the weight matrix varies. 
We note that similar transition has been obtained by others \cite{HM08,YFG13} for a specific parametrization of two-parameter qubit model. 
Here we emphasize that our result is most general and is expressed in terms only of 
the weight matrix and two quantum Fisher information. 

The main result of this paper is summarized in the following theorem 
(The detail of these quantities will be given later): 
Consider a two-parameter qubit model of mixed states, which changes smoothly about variation of the parameter. 
Denote the SLD and RLD Fisher information matrix by $G_\theta$ and $\tilde{G}_\theta$, respectively,  
and define the SLD and RLD CR bounds by 
\begin{align}
C_\theta^S[W]&=\Tr{W G_\theta^{-1}},\\
C_\theta^R[W]&=\Tr{W \Re\tilde{G}_\theta^{-1}} +\Trabs{W\Im\tilde{G}_\theta^{-1}},
\end{align}
respectively. Here, $W$ is a given $2\times2$ positive definite matrix and is called a weight matrix 
($\Trabs{}$ is defined after Eq.~\eqref{zmatrix}.). 
Introduce another quantity by 
\be
C_\theta^Z[W]:=\Tr{W \Re Z_\theta} +\Trabs{W\Im Z_\theta}, 
\ee 
where $Z_\theta$ is a $2\times 2$ hermite matrix defined by 
\be
Z_\theta:=\big[z_\theta^{ij} \big]_{i,j\in\{1,2\}},\quad z_\theta^{ij}:= \tr{\rho_\theta L_\theta^j L_\theta^i}.  
\ee
Here $L_\theta^i$ are a linear combination of the SLD operators; $L_\theta^i=\sum_j g^{ij}_\theta L_{\theta,j}$ 
with $g^{ij}_\theta$ denoting the $(i,j)$ component of the inverse of SLD Fisher information matrix and $L_{\theta,i}$ SLD operators. 
With these definitions, we obtain the following result:
\begin{theorem}\label{thm1}
The Holevo bound for any two-parameter qubit model under the regularity conditions is  
\be \label{eq1}
C_\theta^H[W]=
\begin{cases}
\ds C_\theta^R[W] \qquad \mathrm{if}\ C_\theta^R[W]\ge\frac{C_\theta^Z[W]+C_\theta^S[W]}{2}\\[2ex]
\ds C_\theta^R[W]+S_\theta[W]\qquad \mathrm{otherwise}
\end{cases},
\ee
where the function $S_\theta[W]$, which is nonnegative, is defined by 
\be\label{hsbound}
\ds S_\theta[W]:= \frac {\left[\frac{1}{2}(C_\theta^Z[W]+C_\theta^S[W])- C_\theta^R[W] \right]^2}{C_\theta^Z[W]-C_\theta^R[W]}. 
\ee
\end{theorem}
Note that the condition $C_\theta^R<(C_\theta^Z+C_\theta^S)/2$ implies $C_\theta^Z-C_\theta^R>0$ 
so that $S_\theta[W]$ is well defined. (See the discussion after Eq.~\eqref{zger}.)

The above main result \ref{thm1} sheds several new insights on the quantum parameter estimation problems. 
First note that the form of the Holevo bound changes according to the choice of weight matrices. 
This kind of transition phenomenon has never occurred in the classical case. 
Second surprise is to observe the appearance of the RLD CR bound 
in the generic two-parameter estimation problem. As we will provide in the next section, 
the RLD CR bound has been shown to be important for a special class of 
statistical models, known as a D-invariant model. Here, we explicitly show 
that it also plays a major role for non D-invariant models. 
Last, in many of previous studies parameter estimation problems, 
the precision bound is either expressed in terms of the SLD or RLD Fisher information, 
but the second case of the above expression \eqref{eq1} depends {\it both} on 
the SLD and RLD Fisher information. To see this explicitly, we can rewrite it as
\begin{multline}\label{eq2}
C_\theta^R[W]+S_\theta[W]=
\Tr{WG_\theta^{-1}} \\
+\frac14\ \frac{\left(\Trabs{W\Im\tilde{G}_\theta^{-1}}\right)^2}{\Tr{W(G_\theta^{-1}-\Re\tilde{G}_\theta^{-1})}}. 
\end{multline}
All these findings will be discussed in details together with examples. 

The rest of this paper continues as follows. 
Section \ref{sec2} provides definitions and some of known results for parameter estimation theory 
within the asymptotically unbiased setting. 
In Sec.~\ref{sec3}, a useful tool based on the Bloch vector is introduced and then the above main theorem is proved. 
Discussions on the main theorem are presented in Sec.~\ref{sec4}. 
Section \ref{sec5} gives several examples to illustrate findings of this paper. 
Concluding remarks are listed in Sec.~\ref{sec6}. 
Most of the proofs for lemmas are deferred to Appendix \ref{sec:appb}. 
Supplemental materials are given in Appendix \ref{sec:appc}. 


\section{Preliminaries}\label{sec2}
In this section, we establish definitions and notations used in this paper. 
We then list several known results regarding the Holevo bound 
to make the paper self-contained. 
\subsection{Definitions}\label{sec2-1}
Consider a $d$-dimensional Hilbert space $\cH$ ($d<\infty$) and a $k$-parameter family of quantum states
$\rho_\theta$ on it: 
\be
\cM:=\{\rho_\theta \,|\, \theta=(\theta^1,\theta^2,\dots,\theta^k)\in\Theta\subset\bbr^k \}, 
\ee  
where $\Theta$ is an open subset of $k$-dimensional Euclidean space. 
The family of states $\cM$ is called a {\it quantum statistical model} or simply a {\it model}. 
The model discussed throughout the paper is assumed to satisfy certain regularity conditions 
for the mathematical reasons \cite{comment0}.   
For our purpose, the relevant regularity conditions are: 
i) The state $\rho_\theta$ is faithful, i.e., $\rho_\theta$ is strictly positive. 
ii) It is differentiable with respect to these parameters $\theta$ sufficiently many times. 
iii) The partial derivatives of the state $\del \rho_\theta/\del\theta^i$ are all linearly independent. 
In the rest of this paper, the regularity conditions above are taken for granted unless otherwise stated. 

For a given quantum state $\rho\in\sofh$, we define the SLD and RLD inner products by 
\begin{align} \label{srinn}\nonumber
\langle X,Y\rangle_\rho&:=\frac12\tr{\rho(YX^*+X^* Y)}=\Re \{\tr{\rho YX^*}\}, \\
\langle X,Y\rangle_\rho^+&:=\tr{\rho YX^*}, 
\end{align}
respectively, for any (bounded) linear operators $X,Y$ on $\cH$. 
Here, $*$ denotes the hermite conjugation. 
Given a $k$-parameter model $\cM=\{\rho_\theta\,|\, \theta\in\Theta\subset\bbr^k\}$, 
the SLD operators $L_{\theta, i}$ and RLD  operators $\tilde{L}_{\theta, i}$ are 
formally defined by the solutions to the operator equations:
\begin{align}\label{srdef}\nonumber
\frac{\del}{\del\theta_i}\rho_{\theta}&=\frac12 (\rho_{\theta}L_{\theta,i}+L_{\theta,i}\rho_{\theta}), \\
\frac{\del}{\del\theta_i}\rho_{\theta}&=\rho_{\theta}\tilde{L}_{\theta,i}. 
\end{align}
The SLD Fisher information matrix is defined by \cite{helstrom} 
\begin{align} \label{sldfisher}
G_{\theta}&:= \left[ g_{\theta, ij}\right]_{i,j\in\{1,\dots,k\}} \\ \nonumber
g_{\theta, ij}&:=\sldin{L_{\theta,i}}{L_{\theta,j}}=\tr{\rho_{\theta}\frac12 \big(L_{\theta,i}L_{\theta,j}+L_{\theta,j}L_{\theta,i}  \big)}, 
\end{align}
and the RLD Fisher information is \cite{yl73,holevo}
\begin{align} \nonumber
\tilde{G}_{\theta}&:= \left[ \tilde{g}_{\theta, ij} \right]_{i,j\in\{1,2,\dots,k\}} \\ \label{rldfisher}
\tilde{g}_{\theta, ij}&:=\rldin{\tilde{L}_{\theta,i}}{\tilde{L}_{\theta,j}}=\tr{\rho_{\theta}\tilde{L}_{\theta,j}\tilde{L}_{\theta,i}^*}. 
\end{align}

Define the following linear combinations of 
the SLD and RLD operators: 
\be\nonumber
L_{\theta}^i:= \sum_{j=1}^k(G_\theta^{-1})^{ji}L_{\theta,j}\ \mathrm{and}\ 
\tilde{L}_{\theta}^i:=\sum_{j=1}^k(\tilde{G}_\theta^{-1})^{ji}\tilde{L}_{\theta,j}. 
\ee
From these definitions, the following orthogonality conditions hold.  
\be \label{orthcond}
\sldin{L_{\theta}^i}{L_{\theta,j}}=\delta^i_{\,j},\ \mathrm{and}\ 
\rldin{\tilde{L}_{\theta}^i}{\tilde{L}_{\theta,j}}=\delta^i_{\,j}.
\ee
These operators with upper indices are referred to as 
the {\it SLD and RLD dual operators}, respectively. 

Consider $n$th i.i.d. extension of a given state and we define 
$n$th extended model by 
\be \label{nmodel}
\cM^{(n)}:=\{\rho_\theta^{\otimes n} \,|\, \theta\in\Theta\subset\bbr^k \}.  
\ee
The main objective of quantum statisticians is to perform a measurement 
on the $n$ tensor state $\rho_\theta^{\otimes n}$ and then to make 
an estimate for the value of the parameter $\theta$ based on the measurement outcomes. 
Here measurements are described mathematically by a positive operator-valued measure (POVM) 
and is denoted as $\Pi^{(n)}$. An estimator, which is a purely classical data processing, 
is a (measurable) function taking values on $\Theta$ and is denoted as $\hat{\theta}_{n}$. 
They are
\begin{align*}
\Pi^{(n)}&=\{\Pi_x^{(n)} \,|\, x\in\cX_n,\, \forall \Pi_x^{(n)}\ge0,\,\sum_{x\in\cX_n}\Pi_x^{(n)}=I^{\otimes n}\}, \\
\hat{\theta}^{(n)}&=(\hat{\theta}^1_{n},\hat{\theta}^2_{n},\dots,\hat{\theta}^k_{n}):\,  \cX_n\rightarrow \Theta,
\end{align*}
where $I$ is the identity operator on $\cH$ and we assume that 
POVMs consist of discrete measurement outcomes. For continuous POVMs, we replace the summation by an integration. 
A pair $(\Pi^{(n)},\hat{\theta}_{n})$ is called a {\it quantum estimator} or simply an {\it estimator} 
when is clear from the context and is denoted by $\hat{\Pi}^{(n)}$ \cite{comment1}. 

The performance of a particular estimator can be compared to others based on 
a given figure of merit and then one can seek the ``best" estimator accordingly. 
As there is no universally accepted figure of merit, one should carefully 
adopt a reasonable one depending upon a given situation. For example, 
a specific prior distribution for the parameter $\theta$ is known, 
the Bayesian criterion would be meaningful to find the best Bayesian estimator. 
If one wishes to avoid bad performance of estimators, the min-max criterion 
provides an optimal one that suppresses such cases. 
In this paper, we are interested in analyzing estimation errors at specific point $\theta\in\Theta$, 
that is, the pointwise estimation setting. 
For a given model \eqref{nmodel} and an estimator $\hat{\Pi}^{(n)}=(\Pi^{(n)},\hat{\theta}_{n})$, 
we define a {\it bias} at a point $\theta$ as
\begin{align}\nonumber
b^{(n)}_\theta[\hat{\Pi}^{(n)}]&:=\sum_{x\in\cX_n}(\hat\theta_{n}(x)-\theta)p^{(n)}_\theta(x)=\Eof{\hat\theta_{n}}-\theta,\\
\mathrm{with}&\quad p^{(n)}_\theta(x) :=\tr{\rhon\Pi^{(n)}_x},
\end{align}
where $\Eof{X^{(n)}}$ denotes the expectation value of a random variable $X^{(n)}$ 
with respect to the probability distribution $p^{(n)}_\theta$. 
Note that the bias $b^{(n)}_\theta[\hat\Pi^{(n)}]$ is a $k$-dimensional real vector. 
An important class of estimators when estimating the specific point of the model is 
the locally unbiased estimatior. This is to restrict estimators 
such that the bias vanishes at the true point $\theta$ up to the first order in 
the Taylor expansion. Mathematically, an estimator $\hat{\Pi}^{(n)}$ is called 
{\it locally unbiased} at $\theta$ if 
\be \label{lucondition}
b^{(n)}_\theta[\hat{\Pi}^{(n)}]=0\ \mathrm{and}\ 
\frac{\del}{\del\theta^i} \Eof{\hat\theta_{n}^j}=\delta^{j}_{\,i}, 
\ee
hold at $\theta\in\Theta$. It is known that the Quantum CR bounds 
hold for any locally unbiased estimator \cite{helstrom,holevo}. 

Upon analyzing performance of estimators within the asymptotic regime, 
we should impose some conditions that restrict the class of estimators. 
In statistics, a sequence of an estimator is said (weakly) {\it consistent}, 
if it converges to the true value in probability for every value $\theta\in\Theta$, i.e., 
\be\nonumber
\forall \epsilon>0\quad \lim_{n\to\infty}\mathrm{Pr}\{|\hat\theta_{n}-\theta|>\epsilon \}=0,
\ee
holds for all $\theta\in\Theta$. In this expression, 
$|v|=({\sum v_i^2})^{1/2}$ denotes the standard Euclidean norm 
and the right hand side means that error probability can be made arbitrary small. 
As a good estimator must converge to the true value as $n$ goes to infinity, 
it is reasonable to look for the class of consistent estimators in quantum parameter estimation as well. 
In classical statistics, this condition of consistency alone turns out to be 
weak in order to exclude artificial estimators. There are several approaches to 
handle these problems in the classical case \cite{Vaart}. 
Rather than going into mathematical discussions, 
we simply look for the following class of estimators to avoid such a situation. 
A sequence of estimators $\{\hat{\Pi}^{(n)}\}_{n=1}^\infty$ is called {\it asymptotically unbiased} if it satisfies
\be \label{aucondition}
\lim_{n\to \infty} \sqrt{n}b^{(n)}_\theta[\hat{\Pi}^{(n)}]=0,\ 
\lim_{n\to \infty}\frac{\del}{\del\theta^i} \Eof{\hat\theta_{n}^j}=\delta^{j}_{\,i}, 
\ee
for all $i,j\in\{1,2,\dots,k\}$ and for all $\theta\in\Theta$. That is 
to require the locally unbiased condition \eqref{lucondition} in the $n\to\infty$ limit. 

To quantify estimation errors of a given estimator, we consider the {\it mean-square error} (MSE) matrix defined by 
\begin{align*}
V^{(n)}_\theta[\hat\Pi^{(n)}]&:=\left[ v^{ij}_{\theta,n}[\hat\Pi^{(n)}] \right]_{i,j\in\{1,2,\dots,k\}},\\
v^{ij}_{\theta,n}[\hat\Pi^{(n)}]&:=\sum_{x\in\cX_n} (\hat{\theta}_{n}^i(x)-\theta^i) (\hat{\theta}_{n}^j(x)-\theta^j)p^{(n)}_\theta(x).
\end{align*}
By definition, the MSE matrix is a $k\times k$ real symmetric matrix and it is straightforward 
to show that it is nonnegative. As stated in the introduction, 
we wish to find the best precision bound allowed by the laws of quantum theory and statistics, 
which is achievable in the $n\to\infty$ limit. 
In the classical case, one can directly minimize the MSE matrix as a matrix inequality over the class of asymptotically unbiased estimators 
and to find the lowest MSE error achievable as $n\to\infty$. 
This line of approach does not work in the quantum case. One way to 
tackle this question is to deal with a weighted trace of the MSE matrix, 
which is a scalar quantity, and it is defined by
\be\nonumber
\mathrm{MSE}_\theta^{(n)}[\hat{\Pi}^{(n)}|W]:=\Tr{W V^{(n)}_\theta[\hat{\Pi}^{(n)}]}.  
\ee 
Here the matrix $W$ is called a {\it weight matrix} and can be chosen 
arbitrary as long as it is real and strictly positive. Since the weight matrix is one of the important 
ingredient for our discussion, let us denote the set of all possible weight matrices by 
\be \label{weightset}
\cW:=\{W\in \bbr^{k\times k}\,|\, W>0\}.  
\ee
By changing the weight matrix, one can explore trade-off relations in estimating 
different parameter components $\theta^i$. We note 
that the weight matrix can depend on the value of the estimation parameter $\theta$ as well. 
For example, it can be chosen as the SLD Fisher information matrix. 

Defining these terminologies, we now state the problem: 
For a specific point of a given i.i.d. model $\cM^{(n)}$, 
what is the best sequence of estimators $\{\hat{\Pi}^{(n)} \}_{n=1}^\infty$ 
and the minimum value of the weighted trace of MSE? 
To put it differently, one wishes to find the optimal sequence of estimators 
that minimum of the first order coefficient $C_\theta[W]$ in 
the large $n$ expansion: 
\be
\mathrm{MSE}_\theta^{(n)}[\hat{\Pi}^{(n)}|W]\simeq \frac{C_\theta[W]}{n} +{\cal O}(n^{-2}), 
\ee
i.e., the fastest decaying rate for the MSE. 
Mathematically, we define the CR type bound for the MSE by the following optimization problem:  
\be\nonumber
C_\theta[W]:=\inf_{\{\hat{\Pi}^{(n)}\}\mbox{:a.u.}}
\big\{ \limsup_{n\to\infty}\,n\,\mathrm{MSE}_\theta^{(n)}[\hat{\Pi}^{(n)}|W] \big\}, 
\ee
where the infimum is taken over all possible sequences of estimators $\{\hat{\Pi}^{(n)}\}_{n=1}^\infty$ 
that is asymptotically unbiased (a.u.). 
Note that this bound depends both on the weight matrix $W$ and the model $\rho_\theta$ at $\theta$. 
The symbol $\theta$ appearing in the bound $C_\theta[W]$ represents the model $\rho_\theta$ at $\theta$. 
Unlike the Bayesian or the min-max 
settings mentioned before, we are interested in understanding statistical 
properties of a given parametric model. This would be important 
in particular study of quantum states from geometrical point of view \cite{ANbook}. 

\subsection{The Holevo bound}\label{sec2-2}
To define the Holevo bound, we need some definitions first. 
For a given quantum statistical model on $\cH$, 
denote a $k$ array of hermite operators on $\cH$ by 
$\vec{X}=(X^1,X^2,\dots, X^k)$, $
X^i\in\lofhh$, i.e., $(X^i)^*=X^i$, for all $i=1,2,\dots,k$,
and define the set $\cX_\theta$ by
\begin{multline}\label{Xset}
\cX_\theta:=\{\vec{X}\,|\,\forall i\,X^i\in\lofhh, \forall i\,\tr{\rho_\theta X^i}=0,\\
 \forall i,j\,\tr{\frac{\del \rho_\theta}{\del\theta^i} X^j}=\delta^j_{\,i} \}.
\end{multline}
The Holevo function \cite{nagaoka89} in the quantum estimation theory is defined by
\be \label{hfunction}
h_\theta[\vec{X}|W]:=\Tr{W\mathrm{Re}\,Z_\theta[\vec{X}]}+\Trabs {W\mathrm{Im}\,Z_\theta[\vec{X}] }, 
\ee
where the $k\times k$ hermite matrix $Z_\theta[\vec{X}]$ is
\be \label{zmatrix}
Z_\theta[\vec{X}]:= \big[ \tr{\rho_\theta X^jX^i} \big]_{i,j\in\{1,\dots,k\}},  
\ee
and TrAbs$X$ denotes the sum of the absolute values of $\lambda_j$ 
with $X=T\mathrm{diag}(\lambda_1,\lambda_2,\dots,\lambda_m) T^{-1}$ for some invertible matrix $T$. 
We note the following relation also holds for any anti-symmetric operator $X$: 
\be
\Trabs{WX}=\sum_i|\lambda_i|=\Tr{|W^{1/2}XW^{1/2} |}, 
\ee 
where $|X|=\sqrt{X^* X}$ denotes the absolute value of a linear operator $X$.  

The Holevo bound is defined through the following optimization:
\be \label{hbound}
C_\theta^H[W]:=\min_{\vec{X}\in\cX_\theta}h_\theta[\vec{X}|W]. 
\ee
The derivation of the above optimization is well summarized in Hayashi and Matsumoto \cite{HM08}. 
Holevo showed that this quantity is a bound for the MSE for estimating 
a single copy of the given state under the locally unbiased condition \cite{holevo}: 
\be
\mathrm{MSE}_\theta^{(1)}[\hat{\Pi}|W]\ge C_\theta^H[W], 
\ee
holds for any locally unbiased estimator $\hat{\Pi}$. 
The nontrivial property of the Holevo bound is the additivity \cite{HM08}: 
\be
C_\theta^H[W,\rhon]=n^{-1}C_\theta^H[W,\rho_\theta], 
\ee
where the notation $C_\theta[W,\rhon]$ represents the Holevo bound about $n$th extended model. 

The following theorem establishes that the Holevo bound is the solution to the problem of our interest:
\begin{theorem}\label{thm2}
For a given model satisfying the regularity conditions, $C_\theta[W]=C_\theta^H[W]$ holds for all weight matrices. 
\end{theorem}
There exist several different approaches upon proving the above theorem. 
Hayashi and Matsumoto \cite{HM08} proved the case for a full qubit model first. 
Gu\c{t}\u{a} and Kahn \cite{GK06} introduced a different tool based 
on (strong) quantum local asymptotic normality to prove the qubit case.  
This was further generalized to full models on any finite dimensional Hilbert space \cite{KG09}. 
However, all these proofs depend on a specific parametrization of quantum states. 
More general proof has been recently established by Yamagata, Fujiwara, and Gill \cite{YFG13}. 

This theorem implies that if we choose an optimal sequence of estimators, 
the MSE behaves as 
\be
\Tr{W V^{(n)}_\theta[\hat{\Pi}^{(n)}]}\simeq \frac{C^H_\theta[W]}{n} +{\cal O}(n^{-2}), 
\ee
for sufficiently large $n$. That is the Holevo bound is the fastest 
decaying rate for the MSE. 

Although the Holevo bound stands as an important cornerstone to 
set the fundamental precision bound, the definition \eqref{hbound} 
contains a nontrivial optimization. The main motivation of our work, 
as stated in the introduction, is to perform this optimization explicitly 
for any given model for qubit case. The result shows several 
nontrivial aspects of parameter estimation in quantum domain. 
Before going to present our result, we summarize several known results. 

\subsection{Holevo bound for one-parameter and D-invariant models}\label{sec2-3}
In this subsection, we consider two special cases where analytical forms 
of the Holevo bound are known. 

For a given $k$-parameter model on the Hilbert space $\cH$, 
let us denote SLD and RLD Fisher information matrices by $G_\theta$ and $\tilde{G}_\theta$, 
respectively, Eqs.~(\ref{sldfisher}, \ref{rldfisher}). Define the SLD and RLD CR bounds by
\begin{align}\label{sldbound}
C_\theta^S[W]&:=\Tr{WG_\theta^{-1}},\\ \label{rldbound}
C_\theta^R[W]&:=\Tr{W\Re\tilde{G}_\theta^{-1}}+\Trabs{W{\Im}\tilde{G}_\theta^{-1}},
\end{align} 
respectively. Throughout the paper, we use the notation $\Re\tilde{G}_\theta^{-1}=\Re\{\tilde{G}_\theta^{-1}\}$ 
($\Im\tilde{G}_\theta^{-1}=\Im\{\tilde{G}_\theta^{-1}\}$) 
representing the real (imaginary) part of the inverse matrix of the RLD Fisher information matrix.  
The well-known fact is that the SLD and RLD CR bounds cannot be better than the Holevo bound: 
\begin{lemma}\label{lem2}
For a given model satisfying the regularity conditions, 
the Holevo bound is more informative than the SLD and the RLD CR bound, i.e., 
$C^H_\theta[W]\ge C_\theta^S[W]$ and $C^H_\theta[W]\ge C_\theta^R[W]$ hold for 
an arbitrary weight matrix $W$. 
\end{lemma}
Proof can be found in the original work by Holevo that is summarized in his book \cite{holevo}. 
More compact proof was stated by Nagaoka \cite{nagaoka89}. See also Hayashi and Matsumoto \cite{HM08}. 

\subsubsection{One-parameter model}
When the number of parameters is one, the problem can be reduced significantly. 
In this case, there cannot be any imaginary part for the matrix \eqref{zmatrix} 
and thus the minimization is reduced to minimizing the MSE itself. 
\begin{theorem}\label{thm3}
For any one-parameter model, the Holevo bound coincides with the SLD CR bound, i.e., 
\be \label{1bound}
C_\theta^H=\frac{1}{g_\theta}, 
\ee
holds for all $\theta\in\Theta$ where $g_\theta$ is the SLD Fisher information at $\theta$. 
\end{theorem}
Note that there is no weight matrix since we are dealing with a scalar MSE for the one-parameter case. 
Importantly, there is no gain from collective POVMs for one-parameter models. 
Existence of a POVM whose MSE is equal to this bound is discussed independently 
by several authors \cite{yang,bc94,nagaoka87}. 

\subsubsection{D-invariant model} \label{sec2-3-2}
Consider an arbitrary $k$-parameter model $\cM$ and let $L_{\theta,i}$ ($i=1,2,\dots,k$) be 
the SLD operators at $\theta$. The linear span of SLD operators with real coefficients is called 
the {\it SLD tangent space} of the model at $\theta$:
\be
T_\theta(\cM):=\mathrm{span}_\bbr\{L_{\theta,1},L_{\theta,2},\dots,L_{\theta,k} \}. 
\ee
Any elements of the SLD tangent space, $X\in T_\theta(\cM)$, satisfy $\tr{\rho_\theta X}=0$ 
and it is not difficult to see that the space $T_\theta(\cM)$ is essentially a real vector 
space with the dimension $k$. 
Holevo introduced a super-operator $\cD{\rho}$, called a {\it commutation operator}, as follows.
Given a state $\rho$ on $\cH$, let $\lofh$ be the set of linear operators on $\cH$, then 
$\cD{\rho}$ is a map from $\lofh$ to itself defined through the following equation: 
\be
\langle Y,\cD{\rho}(X)\rangle_\rho=[Y,\,X]_\rho,\quad \forall X,Y\in\lofh. 
\ee
Here,  $\langle X,Y\rangle_\rho=\Re\,\tr{\rho X^* Y}$ 
is the SLD inner product and $[X,\,Y]_\rho=\tr{\rho[X^*,\,Y]}/(2\I)$ is a sesqui-linear form. 
(Here, the definition is different from the original one by a factor.)
When considering a parametric model, we denote $\cD{\rho_\theta}=\cD{\theta}$ for simplicity. 
We say that a model is {\it D-invariant} at $\theta$ if the SLD tangent space at $\theta$ is invariant 
under the action of the commutation operator. Mathematically, this definition is expressed as 
\be\nonumber
\mbox{$\cM$ is D-invariant at $\theta$}\Def\cD{\theta}(X)\in T_\theta(\cM),\ \forall X\in T_\theta(\cM). 
\ee
When a model is D-invariant for all $\theta\in\Theta$, we say the model $\cM$ is globally D-invariant. 
Lemma \ref{lem4} in Appendix \ref{sec:appc-1} characterizes equivalent conditions for D-invariant models. 

From the definitions of two inner products \eqref{srinn} and the commutation operator, the relationship
\begin{multline}
\langle X,Y\rangle_\rho^+=\langle X, (I+\I \cD{\rho})(Y)\rangle_\rho\\
\Leftrightarrow \langle X,\cD{\rho}(Y)\rangle_\rho=-\I (\langle X,Y\rangle_\rho^+-\langle X,Y\rangle_\rho), 
\end{multline}
holds for all linear operators $X,Y$ on $\cH$. 
For a given model, another important relation 
\be
\tr{\del_i\rho_\theta X}=\sldin{X}{L_{\theta,i}}=\rldin{X}{\tilde{L}_{\theta,i}},
\ee
holds for $\forall X\in\lofh$. Combining them gives 
$\sldin{X}{L_{\theta,i}}=\sldin{X}{(I+\I\cD{\theta})(\tilde{L}_{\theta,i})}$, 
and hence we obtain 
\be \label{app-0}
L_{\theta,i}=(I+\I\cD{\theta})(\tilde{L}_{\theta,i}). 
\ee
Two more useful relations are 
\begin{align} \label{app1}
\sldin{L_{\theta}^i}{\cD{\theta}(L_{\theta}^j)}&=\Im z_\theta^{ij},\\ \label{app2}
\rldin{\tilde{L}_{\theta}^i}{\cD{\theta}({L}_{\theta}^j)}&=-\I(\tilde{g}_\theta^{ij}-g_\theta^{ij}), 
\end{align}
which can be checked directly from the definitions.

It is well known that the Holevo bound gets simplified significantly 
if the model is D-invariant \cite{holevo}.  
Importantly, D-invariant model enjoys the following proposition, which is due to Holevo: 
\begin{proposition} \label{propDinv}
Let $L_\theta^i$ ($\tilde{L}_\theta^i$) be the SLD (RLD) dual operator 
and $G_\theta$ ($\tilde{G}_\theta$) be the SLD (RLD) Fisher information matrix, respectively. 
Define a $k\times k$ hermite matrix 
by $Z_\theta=[\rldin{L_\theta^i}{L_\theta^j}]_{i,j\in\{1,\dots,k\}}$. 
When the model is D-invariant at $\theta$, $Z_\theta=\tilde{G}_\theta^{-1}$ holds at $\theta$ 
and further the Holevo bound is expressed as 
\be
\forall W\in\cW\ C_\theta^H[W]=C_\theta^R[W] =h_\theta[\vec{L}|W]. 
\ee
\end{proposition}
This statement can be proven in several different manners \cite{holevo,HM08,YFG13}. 

In passing, we note that the expression $h_\theta[\vec{L}|W]$ in the above proposition 
is also expressed as $h_\theta[\vec{L}|W]=\Tr{W \Re Z_\theta} +\Trabs{W\Im Z_\theta}$ in terms of 
the matrix $Z_\theta$. When the model is not D-invariant, $h_\theta[\vec{L}|W]$ does not 
seem to play any important role. This is because the quantity 
\be \label{grldbound}
C_\theta^Z[W]:=\Tr{W \Re Z_\theta} +\Trabs{W\Im Z_\theta}, 
\ee
is always greater or equal to the Holevo bound, i.e., $C_\theta^Z[W]\ge C_\theta^H[W] $ for all weight matrices. 
Nevertheless, as will be shown in this paper, this is an important quantity 
and we call it as the {\it D-invariant bound} in our discussion.  
We note that this quantity \eqref{grldbound} was also named as the {\it generalized RLD CR bound} 
by Fujiwara and Nagaoka \cite{FN99} in the following sense. 
When a model fails to satisfy some of the regularity conditions, the RLD operators do not exist always. 
Even in this case, when the model is D-invariant, then the above bound \eqref{grldbound} is well defined 
and provides the achievable bound for a certain class of models, known as the coherent model \cite{FN99}. 

Another remark regarding this proposition is that the converse statement 
also holds. 
\begin{theorem}\label{thmDinv}
For any $k$-parameter model $\cM$ on any dimensional Hilbert space under the regularity conditions, 
the following equivalence holds:
\begin{align}\nonumber
&\mbox{$\cM$ is D-invariant at $\theta$.}\\ \nonumber
\Leftrightarrow&\ \forall W\in\cW\ C_\theta^H[W]=C_\theta^R[W].\\  
\Leftrightarrow&\ \forall W\in\cW\ C_\theta^H[W]=C_\theta^Z[W]. 
\end{align}
\end{theorem}
This equivalence for the D-invariant model might have been 
known for some experts, but it was not stated explicitly in literature to our knowledge \cite{nagaokaseminar}. 
Sketch of proof is given in Appendix \ref{sec:appc-2} for the sake of reader's convenience. 
 
We remark that the Holevo bound for a general model, which is not D-invariant, exhibits 
a gap among $C_\theta^Z[W]$, $C_\theta^R[W]$, and $C_\theta^S[W]$. The following relation holds in general: 
\be
C_\theta^Z[W]\ge C_\theta^H[W]\ge\max\{C_\theta^S[W],C_\theta^R[W] \},
\ee
for all weight matrices $W$. 
From this general inequality, it is clear that the condition of D-invariance shrinks 
the gap between $C_\theta^Z[W]$ and $C_\theta^R[W]$ to zero. 
The Holevo bound then coincides with the RLD and D-invariant bounds. 

\section{The Holevo bound for qubit models} \label{sec3}
In this section we consider a model for quantum two-level system, a qubit model. 
For mixed-state models, where parametric states are rank-2 for all $\theta\in\Theta$, 
possible numbers of the parameters are from one to three. 
As stated in Sec.~\ref{sec2-3}, the Holevo bound for one-parameter qubit model is solved 
and is given by Theorem \ref{thm3}. When the number of parameters is equal to 
three, on the other hand, it is easy to show that the model becomes D-invariant as follows.  
Since three SLD operators are linearly independent, they expand 
the set of all linear operators $X$ satisfying the condition $\tr{\rho_\theta X}=0$. 
In other words, the SLD tangent space is same as this space 
and hence the SLD tangent space is always D-invariant. 
In this case, the Holevo bound is given as Eq.~\eqref{rldbound}. 
Therefore, the two-parameter case needs to be solved explicitly. 
In the following, we consider two-parameter qubit models of mixed states only. 
Further, the regularity conditions mentioned before are assumed throughout our discussion. 

Upon performing this optimization to derive an explicit formula for the Holevo bound, 
it is convenient to utilize the Bloch-vector formalism. 
A similar technique has been used by Watanabe, where all operators are expanded 
in terms of a basis of Lie algebras \cite{watanabeD}. 
In the next subsection, we present 
necessary machinery and then solve the two-parameter case. 

\subsection{Bloch-vector formalism for qubit estimation problem}\label{sec3-1}
In this subsection, we present a formalism in which SLD operators are 
represented by three-dimensional real vector. This is motivated by the well-known 
one-to-one mapping between a given qubit state and three-dimensional real vector.  
Thus, any qubit model can be represented by a family of three-dimensional real vectors as
\be \label{qbmodel}
\cM_{\cB}=\left\{\v{s}_{\theta}=(s_\theta^1,s_\theta^2,s_\theta^3 )\in\cB\, |\, \theta\in\Theta \right\},
\ee
with $\cB=\{ x\in\bbr^3\,|\,|x| <1\}$ the interior of the Bloch ball. 
To simplify notations, we define the standard inner product and 
the outer product for three-dimensional complex vectors by 
\[
\vecin{\v a}{\v b}=\sum_{i=1,2,3}\bar{a}_i b_i,\quad
\ket{\v a}\bra{\v b}=\Big[a_i \bar{b}_j \Big]_{i,j\in\{1,2,3\}},
\]
respectively, where $\bar{a}$ denotes the complex conjugation of $a$. 
The outer product is a $3\times3$ matrix whose action onto a vector 
$\v c\in\bbc^3$ is $\ket{\v a}\bra{\v b} {\v c}=\vecin{\v b}{\v c}{\v a}$.

We first observe that the one-to-one correspondence between 
the SLD operator and a four-dimensional vector when is expanded 
in terms of the basis $\{ I, \sigma_1,\sigma_2,\sigma_3\}$ with 
$\sigma_j$ usual Pauli spin matrices for spin-1/2 particles. 
Since the SLD operators belong to the SLD tangent space, 
the relation $\tr{\rho_\theta L_{\theta,i}}=\sldin{I}{L_{\theta,i}}=0$ holds, 
i.e., they are orthogonal to the identity operator with respect 
to the SLD inner product.  
This leads to the following constraint:  
\be
\tr{L_{\theta,i}}=-\vecin{\stheta}{ \v{\ell}_{\theta,i}},
\ee
where $\v\ell_{\theta,i}=(\ell^1_{\theta,i},\ell^2_{\theta,i},\ell^3_{\theta,i})^T$ with 
$\ell^j_{\theta,i}=\tr{ \sigma_j L_{\theta,i}}$ is a three-dimensional real vector. 
Thus, we have a one-to-one mapping from the SLD operator $L_{\theta,i}$ 
to the three-dimensional real vector $\v\ell_{\theta,i}$. 
The vector $\v\ell_{\theta,i}$ shall be referred to as the {\it SLD Bloch vector} in this paper.  

It is straightforward to solve the operator equation \eqref{srdef}, 
which defines the SLD operators, and the SLD Bloch vector is
\be
\v{\ell}_{\theta,i}= \del_i\stheta+\frac{ \vecin{\del_i \stheta}{\stheta}}{1-s^2_\theta}\stheta, 
\ee
where $s_\theta=|\stheta|$ denotes the length of the Bloch vector 
and $ \del_i=\del/\del\theta^i$ is the $i$th partial derivative. 
To proceed further, we find it convenient to introduce a $3\times 3$ matrix: 
\be \label{Qop1}
Q_{\theta}:=\openone+\frac{\ket{\stheta}\bra{\stheta}}{1-s_\theta^2}, 
\ee
with $\openone$ the identity matrix acting on the three-dimensional vector space $\bbc^3$. 
It follows from the definition that $Q_{\theta}$ is a real and positive matrix with eigenvalues 
$1,1,(1-s^2_\theta)^{-1}$ and its inverse is 
\be\label{Qop2}
Q_{\theta}^{-1}=\openone-\ket{\stheta}\bra{\stheta}.   
\ee
The SLD Bloch vector is then expressed as
\be
\v{\ell}_{\theta,i}=Q_\theta \del_i \stheta\quad
(\Leftrightarrow\ Q^{-1}_\theta\v{\ell}_{\theta,i}= \del_i \stheta). 
\ee
The $(i,j)$ component of the SLD Fisher information is 
\be \label{sldqubit}
g_{\theta,ij}=\vecin{\v{\ell}_{\theta,i}}{Q_\theta^{-1}\v{\ell}_{\theta,j}}=
 \vecin{ \del_i \stheta}{Q_\theta  \del_j \stheta}. 
\ee
Let $g^{ij}_\theta=(G_\theta^{-1})_{ij}$ be the $(i,j)$ component of the inverse SLD Fisher matrix 
and we define the SLD Bloch dual vector $\v{\ell}_{\theta}^i$ by 
\be\label{sldcot}
\v{\ell}_{\theta}^i=\sum_{j}g^{ji}_\theta\v{\ell}_{\theta,j}, 
\ee
then, the following orthogonality condition holds:
\be
\vecin{\v{\ell}_{\theta}^i}{Q_\theta^{-1}\v{\ell}_{\theta,j}}=\delta^i_{j}, 
\ee
which corresponds to Eq.~\eqref{orthcond}. 
The inverse of SLD Fisher information matrix is also expressed as 
\be \label{sldQinv}
g^{ij}_\theta=\vecin{\v{\ell}_{\theta}^i}{Q_\theta^{-1}\v{\ell}_{\theta}^j}. 
\ee

The same line of arguments holds for RLD operators and RLD Fisher information. 
The only difference is here is that the RLD Bloch vector becomes complex in general. 
Define a $3\times 3$ complex matrix: 
\be \label{Qtop1}
\tilde{Q}_{\theta}:=\frac{1}{1-s_\theta^2}\left(\openone-\I F_\theta\right), 
\ee
where $(F_\theta)_{ij}:= \sum_{k}\epsilon_{i k j}s_{\theta,k} $ with $\epsilon_{i k j}$ 
the completely antisymmetric tensor. The action of $F_\theta$ is to give the 
exterior product of two vectors, i.e., $F_\theta \v a=\stheta\times \v a$ for $\v a\in\bbc^3$. 
From this definition, $\tilde{Q}_{\theta}$ is also strictly positive and its inverse is given by 
\be\label{Qtop2}
\tilde{Q}_{\theta}^{-1}=\openone-\ket{\stheta}\bra{\stheta}+\I F_\theta=Q_\theta^{-1}+\I F_\theta.   
\ee
The RLD Bloch vector is
\be
\tilde{\v\ell}_{\theta,i}=\tilde{Q}_\theta \del_i \stheta\quad
(\Leftrightarrow\ \tilde{Q}^{-1}_\theta\tilde{\v\ell}_{\theta,i}= \del_i \stheta), 
\ee
and the RLD Fisher information matrix is
\be\label{rldqubit}
\tilde{g}_{\theta,ij}=\vecin{\tilde{\v\ell}_{\theta,i}}{\tilde{Q}_\theta^{-1}\tilde{\v\ell}_{\theta,j}}=
 \vecin{ \del_i \stheta}{\tilde{Q}_\theta  \del_j \stheta}. 
\ee
Define the RLD Bloch dual vector by
\be\label{rldcot}
\tilde{\v\ell}_{\theta}^i=\sum_{j}\tilde{g}^{ji}_\theta\tilde{\v\ell}_{\theta,j}, 
\ee
then we have 
\be
\vecin{\tilde{\v\ell}_{\theta}^i}{\tilde{Q}_\theta^{-1}\tilde{\v\ell}_{\theta,j}}=\delta^i_{j}. 
\ee

Other useful relations are listed below without detail calculations. 
First, there is a one-to-one correspondence between SLD and RLD Bloch vector. This is 
given by
\be
\v{\ell}_{\theta,i}=\left(\openone+\I F_\theta \right) \tilde{\v\ell}_{\theta,i}\quad
(\Leftrightarrow\  \tilde{\v\ell}_{\theta,i}=\left(Q_\theta^{-1}-\I F_\theta \right) \v{\ell}_{\theta,i}).  
\ee

Second, the vector $\v{\ell}_{\theta,i}-\tilde{\v\ell}_{\theta,i}$ is orthogonal to the Bloch vector $\stheta$. 
Defining 
\be\label{gamma}
\gamma_{\theta,i}:=\vecin{\stheta}{\v{\ell}_{\theta,i}},\quad
\tilde\gamma_{\theta,i}:=\vecin{\stheta}{\tilde{\v\ell}_{\theta,i}}, 
\ee
this is expressed as
\be
\gamma_{\theta,i}=\tilde\gamma_{\theta,i}. 
\ee

Third, the SLD Fisher information and real part of the RLD Fisher information 
are related by
\be \label{sldQrld}
\frac{1}{1-s_\theta^2}\,{g_{\theta,ij}}-\Re\tilde{g}_{\theta,ij}=\gamma_{\theta,i}\gamma_{\theta,j}.  
\ee
In other words, the matrix $(1-s_\theta^2)^{-1}\,G_{\theta}-\Re\tilde{G}_{\theta}$ is rank one.

\subsection{Two-parameter qubit model}\label{sec3-2}
In this subsection, we consider an arbitrary two-parameter qubit model, that is 
the parameter to be estimated is $\theta=(\theta_1,\theta_2)\in\Theta$. 
In order to derive an explicit expression for the Holevo bound, 
we first rewrite the Holevo bound \eqref{hbound} in terms of the Bloch vectors. 

A linear operator which satisfies $\tr{\rho_\theta X^i}=0$ can be expressed as 
\be
X^i=-\vecin{\stheta}{\v x^i}I+{\v x^i}\cdot{\v \sigma} .  
\ee 
Let $T_{\theta,i}^{\bot}=\{\v x\in\bbr^3\,|\, \vecin{x}{\del_i \v \stheta}=0 \}$ 
be the orthogonal space to the $i$th derivative of the Bloch vector and 
an element of the set appearing in the definition \eqref{Xset} takes the form of 
$\vec{X}=(X^1,X^2)$ with 
\begin{align}\nonumber \label{h2set}
X^i&=-\vecin{\stheta}{\v x^i}I+{\v x^i}\cdot{\v \sigma},\\
\v{x}^i&\in T_{\theta,j}^{\bot}\mbox{ for $j\neq i$\quad and}\quad \vecin{\v{x}^i}{\del_i\stheta}=1.
\end{align}
Thus, the set of operators $\vec{X}$ can be mapped to 
the set of vectors $\vec{x}=(\v x^1,\v x^2)^T\in\bbr^6$. 
Using this form of Bloch vector representation, the $(i,j)$ component of the $Z_\theta[\vec{x}]$ matrix 
and the Holevo function read
\begin{align}\label{h2function}
z_\theta^{ij}[\vec x]&= \vecin{\v x^i}{\tilde{Q}_\theta^{-1}\v x^j}, \\ \nonumber
h_\theta[\vec{x}|W]&=\sum_{i,j=1}^2 \left[ w_{ij}\vecin{\v x^i}{Q_\theta^{-1}\v x^j} 
+\sqrt{\det W}\big| \vecin{\v x^i}{F_\theta\v x^j} \big| \right],
\end{align}
for a given $2\times2$ weight matrix $W=[w_{ij}]_{i,j\in\{1,2\}}$. 

We note that the Holevo function \eqref{h2function} is a 
quadratic function of the six-dimensional vector $\vec{x}$.
The minimization of this function under the constraints \eqref{h2set} can be 
handled with a standard procedure. The only point needs to be
taken is that the function is not differentiable for all points. 
Since the number of free variables for the optimization is 
$6-4$(the number of independent constraints)$=2$, 
we take the following substitution: 
\be
\v x^i=\v\ell_\theta^i+\xi^i \v\ell_{\theta}^{\bot}, 
\ee
where $\ell_{\theta}^{\bot}=\del_1\stheta\times\del_2\stheta$ is a 
vector orthogonal to both $\del_1\stheta$ and $\del_2\stheta$ 
and $\vec\xi=(\xi^1,\xi^2)^T\in\bbr^2$ is a free variable without any constraint.   
With this substitution, the Holevo function is significantly simplified as follows. 
\begin{lemma}\label{lem-holevo2}
For a two-parameter qubit model, the Holevo bound takes 
the following minimization form without any constraint: 
\be
C_\theta^H[W]=\min_{\vec{\xi}\in\bbr^2}h_\theta[\vec{\xi}|W],  
\ee
where the function $h_\theta[\vec{\xi}|W]$ is defined by 
\begin{multline} \label{h3function}
h_\theta[\vec{\xi}|W]=\Tr{WG_\theta^{-1}}
+\vecin{\v\ell_{\theta}^\bot}{{Q}_{\theta}^{-1}\v\ell_{\theta}^\bot} 
(\vec\xi|W\vec\xi)\\
+2\sqrt{\det W}\left| \vecin{\v \ell_\theta^1}{F_\theta\v \ell_\theta^2}+(1-s_\theta^2)(\vec{\gamma}_\theta|\vec{\xi}) \right|. 
\end{multline}
\end{lemma}
In this expression, we introduce the standard inner product for two-dimensional real vector space by 
$(\vec{a}|\vec{b})=a_1b_1+a_2b_2$ 
and $\vec{\gamma}_\theta=(\gamma_{\theta,1},\gamma_{\theta,2})^T$ is given by Eq.~\eqref{gamma}.  
The derivation for this lemma is given in Appendix \ref{sec:appb-0}. 

\subsection{Main result}\label{sec3-3}
In the following, we carry out the above optimization to derive the main result of this paper, Theorem \ref{thm1}. 
We first list several definitions and lemmas. 
\begin{definition}
For a given two-parameter qubit mixed-state model, the SLD CR, RLD CR, and D-invariant bounds are defined by
\begin{align}
C_\theta^S[W]&=\Tr{W G_\theta^{-1}},\\ \nonumber
C_\theta^R[W]&=\Tr{W \Re\tilde{G}_\theta^{-1}} +\Trabs{W\Im\tilde{G}_\theta^{-1}},\\ \nonumber
C_\theta^Z[W]&=\Tr{W \Re Z_\theta} +\Trabs{W\Im Z_\theta}, 
\end{align}
respectively, where $G_\theta$ is the SLD Fisher information matrix, $\tilde{G}_\theta$ is the RLD Fisher information matrix, 
and $Z_\theta:=\big[z_\theta^{ij} \big]_{i,j\in\{1,2\}}$ with $z_\theta^{ij}:= \tr{\rho_\theta L_\theta^j L_\theta^i}$ as before. 
\end{definition}
\begin{lemma}\label{lem7}
For any two-parameter qubit mixed-state model, the following relations hold: 
\begin{enumerate}
\item $\ds\vecin{\v\ell_{\theta}^\bot}{{Q}_{\theta}^{-1}\v\ell_{\theta}^\bot} = 
(1-s_\theta^2)\det\, G_\theta=(1-s_\theta^2)^2\det\, \tilde{G}_\theta$.
\item $\ds2\sqrt{\det W}\left| \vecin{\v \ell_\theta^1}{F_\theta\v \ell_\theta^2}\right|=\Trabs{W\Im\tilde{G}_\theta^{-1}}=\Trabs{W\Im Z_\theta}$. 
\item $\ds (\vec{\gamma}_\theta|W^{-1}\vec{\gamma}_\theta)=
\frac{\det W^{-1}G_\theta}{1-s_\theta^2}\left(C_\theta^Z[W]-C_\theta^R[W] \right)$. 
\end{enumerate}
\end{lemma}

\begin{lemma}\label{lem8}
For any two-parameter qubit model $\cM$, the following conditions are equivalent. 
\begin{enumerate}
\item $\cM$ is D-invariant at $\theta$.
\item $\Re \tilde{G}_\theta^{-1}=G_\theta^{-1}$ at $\theta$. 
\item $\gamma_{\theta,1}=\gamma_{\theta,2}=0$ at $\theta$.
\end{enumerate}
Furthermore, we have the following equivalent characterization for global D-invariance. 
\begin{enumerate}
\item[4.] $\cM$ is globally D-invariant. 
\item[5.] $\Re \tilde{G}_\theta^{-1}=G_\theta^{-1}$ for all $\theta\in \Theta$. 
\item[6.] $|\stheta|$ is independent of $\theta$. 
\end{enumerate}
\end{lemma}
Three remarks regarding the above lemmas are listed: 
First, imaginary parts of the inverse of the RLD Fisher information matrix 
and the $Z_\theta$ matrix are always identical for two-parameter qubit mixed-state models, 
i.e., $\Im \tilde{G}_\theta^{-1}=\Im Z_\theta$, see proof in the appendix. 
Second, if a model expressed as the Bloch vector contains the origin $(0,0,0)$, the model is always D-invariant 
at this point. This is because the condition $\gamma_{\theta,i}=\vecin{\stheta}{\v{\ell}_{\theta,i}}=0$ 
is met at $\stheta=(0,0,0)$. 
Last, a globally D-invariant model is possible if and only if 
the state is generated by some unitary transformation. 
This is because the condition \ref{lem8}-6 in Lemma \ref{lem8} is equivalent 
to preservation of the length of the Bloch vector. 
 
Finally, we need the following lemma for the optimization: 
\begin{lemma}\label{lem9}
For a given positive $2\times 2$ matrix $A$, a real vector $\vec b\in\bbr^2$, and a real number $c$, 
the minimum of the function
\be\nonumber
f(\vec{\xi})=(\vec{\xi}|A\vec\xi)+2\left|(\vec{b}|\vec\xi)+c\right|, 
\ee
is given by 
\be\nonumber
\min_{\vec\xi\in\bbr^2}f(\vec{\xi})=
\begin{cases}
\ds 2|c|-(\vec{b}|A^{-1}\vec{b}) \quad \mathrm{if}\ |c|\ge(\vec{b}|A^{-1}\vec{b}) \\[2ex]
\ds \frac{|c|^2}{(\vec{b}|A^{-1}\vec{b})}\qquad \quad \mathrm{if}\ |c|<(\vec{b}|A^{-1}\vec{b}) 
\end{cases},
\ee
where the minimum is attained by 
\be\nonumber
\vec{\xi}_*=
\begin{cases}
\ds-\mathrm{sign}(c)A^{-1}\vec{b}\quad\quad\mathrm{if}\ |c|\ge(\vec{b}|A^{-1}\vec{b})\\[2ex]
\ds-\frac{c}{(\vec{b}|A^{-1}\vec{b})}A^{-1}\vec{b}\quad\mathrm{if}\ |c|<(\vec{b}|A^{-1}\vec{b})
\end{cases},
\ee
where $\mathrm{sign}(c)$ is the sign of $c$. 
\end{lemma}
Proofs for the above three lemmas are given in Appendix \ref{sec:appb-2}-\ref{sec:appb-4}. 

\subsubsection{Proof for Theorem \ref{thm1}}
We now prove Theorem \ref{thm1}. 
From the expression of the Holevo function \eqref{h3function}, we can apply Lemma \ref{lem9} 
by identifing 
\begin{align}
A&=\vecin{\v\ell_{\theta}^\bot}{{Q}_{\theta}^{-1}\v\ell_{\theta}^\bot}  W,\\
 \vec b&=(1-s_\theta^2)\sqrt{\det W}\vec{\gamma}_\theta,\\
c&=\sqrt{\det W} \vecin{\v \ell_\theta^1}{F_\theta\v \ell_\theta^2}. 
\end{align}
We need to evaluate $(\vec{b}|A^{-1}\vec{b})$ and $|c|$ and they are calculated as follows. 
\begin{align*} 
(\vec{b}|A^{-1}\vec{b})&= 
(1-s_\theta^2)^2\det{W} (\vec{\gamma}_\theta |(\vecin{\v\ell_{\theta}^\bot}{{Q}_{\theta}^{-1}\v\ell_{\theta}^\bot}  W)^{-1}\vec{\gamma}_\theta)\\
&=\frac{(1-s_\theta^2)^2\det{W} }{ \vecin{\v\ell_{\theta}^\bot}{{Q}_{\theta}^{-1}\v\ell_{\theta}^\bot} }  (\vec{\gamma}_\theta |W^{-1}\vec{\gamma}_\theta)\\
&=C_\theta^Z[W]-C_\theta^R[W], 
\end{align*}
where Lemma \ref{lem7}-1 and \ref{lem7}-3 are used to get the last line. Lemma \ref{lem7}-2 
immediately gives 
\be
2|c|=\Trabs{W\Im Z_\theta}=C_\theta^Z[W]-C_\theta^S[W]. 
\ee
Therefore, we obtain if 
\be \nonumber
|c|\ge(\vec{b}|A^{-1}\vec{b})\Leftrightarrow 
C_\theta^R[W]\ge\frac12 (C_\theta^Z[W]+C_\theta^S[W])
\ee 
is satisfied, the Holevo bound is
\begin{align}
C_\theta^H[W]&=\nonumber
\ds C_\theta^S[W]+ (C_\theta^Z[W]-C_\theta^S[W])- (C_\theta^Z[W]-C_\theta^R[W])\\
&= C_\theta^R[W]. 
\end{align}

If $|c|<(\vec{b}|A^{-1}\vec{b})$ $\Leftrightarrow$ $C_\theta^R[W]<({C_\theta^Z[W]+C_\theta^S[W]})/{2}$ is satisfied, 
on the other hand, the Holevo bound takes the following form: 
\begin{align} \label {sfunction}
C_\theta^H[W]
&=C_\theta^S[W]+\frac{\left[(C_\theta^Z[W]-C_\theta^S[W])/2\right]^2}{C_\theta^Z[W]-C_\theta^R[W]} \nonumber\\
&=C_\theta^R[W]\nonumber\\ 
+&\frac{1}{C_\theta^Z[W]-C_\theta^R[W]} \left(\frac{C_\theta^Z[W]+C_\theta^S[W]}{2}- C_\theta^R[W] \right)^2 \nonumber\\
&=C_\theta^R[W]+S_\theta[W], 
\end{align}
where the function $S_\theta[W]$ is defined in Eq.~\eqref{hsbound}. 
This proves the theorem. $\square$

We remark that from Lemma \ref{lem7} and the positivity of $W$ we always have the relation
\be \label{zger}
C_\theta^Z[W]\ge C_\theta^R[W], 
\ee
and the equality if and only if $ (\vec{\gamma}_\theta|W^{-1}\vec{\gamma}_\theta)=0$ 
$\Leftrightarrow$ $\cM$ is D-invariant at $\theta$ by Lemma \ref{lem8}-3. 
Note if $\cM$ is D-invariant at $\theta$ ($\gamma_{\theta,1}=\gamma_{\theta,2}=0$), 
the condition $C_\theta^R[W]<({C_\theta^Z[W]+C_\theta^S[W]})/{2}$ 
$\Leftrightarrow$ $0\le \mathrm{TrAbs}\{W\Im \tilde{G}_\theta^{-1}\}/2<C_\theta^Z[W]-C_\theta^R[W] $ cannot be satisfied. 
Thus, the obtained Holevo is well defined for all $\theta$ and for arbitrary weight matrix $W$. 

The optimal set of hermite operators attaining the Holevo bound can be given by Lemma \ref{lem9} as follows. 
Define an hermite matrix by
\be
L_\theta^\bot:= -\vecin{\stheta}{\v{\ell}_\theta^\bot}I+{\v{\ell}_\theta^\bot}\cdot{\v \sigma}, 
\ee 
and the function $\zeta(W)$ by 
\begin{multline}\nonumber
\zeta_\theta(W):=\det W^{-1/2}G_\theta^{-1}\\
\times\begin{cases}
 \mathrm{sign}(\Im z_\theta^{12})\quad \mathrm{if}\ C_\theta^R[W]\ge\frac12(C_\theta^Z[W]+C_\theta^S[W])\\[1ex]
\frac{W^{1/2}\Im z_\theta^{12}}{C_\theta^Z[W]-C_\theta^R[W] }\quad \mathrm{otherwise} 
\end{cases}. 
\end{multline}
Then, we have $\vec{X}_*=(X_*^1,X_*^2)=\mathrm{arg}\min h_\theta[\vec{X}|W]$ as  
\begin{align} \nonumber
X_*^1&=L_\theta^1+\zeta_\theta (w_{22}\gamma_{\theta,1}-w_{12}\gamma_{\theta,2}) L_\theta^\bot,\\
X_*^2&=L_\theta^2+\zeta_\theta (-w_{21}\gamma_{\theta,1}+w_{11}\gamma_{\theta,2}) L_\theta^\bot. 
\end{align}

Before moving to discussion and consequence of the main result, we present the following two-alternative expressions.  
Define a function, 
\be
H(x):=\begin{cases}
2|x|-1\quad \mbox{if $|x|\ge1$}\\[1ex]
x^2\qquad \mbox{if $|x|<1$}
\end{cases}, 
\ee 
which is continuous and the first derivative is also continuous for all $x\in\bbr$. 
Then, the Holevo bound is written in a unified manner: 
\begin{multline}\label{h4bound}
C_\theta^H[W]=C_\theta^S[W]+(C_\theta^Z[W]-C_\theta^R[W]) \\
\times H\left(\frac{(C_\theta^Z[W]-C_\theta^S[W])/2}{C_\theta^Z[W]-C_\theta^R[W]}\right). 
\end{multline}
This expression needs a special care when $C_\theta^Z[W]-C_\theta^R[W]\to 0$. 
This case should be understood as the limit $\lim_{a\to0}aH(b/a)=2|b|$. 

The other expression shown in Eq.~\eqref{eq2} follows from the first line of 
Eq.~\eqref{sfunction} by noting $C_\theta^Z[W]-C_\theta^R[W]= \mathrm{Tr}\{W(G_\theta^{-1}-\Re\tilde{G}_\theta^{-1})\}$ 
and $C_\theta^Z[W]-C_\theta^S[W]=\mathrm{TrAbs}\{W\Im \tilde{G}_\theta^{-1}\}$.


\section{Discussion on Theorem \ref{thm1}}\label{sec4} 
In this section, we shall discuss the consequences of Theorem \ref{thm1}. 
This brings several important findings of our paper. 
First is two conditions that characterize special classes of qubit models. 
Second is a transition in the structure of the Holevo bound depending on the choice of the weight matrix. 

\subsection{Necessary and sufficient conditions for special cases}\label{sec4-1} 
The general formula for the Holevo bound for any two-parameter model is rather unexpected in the following sense. 
First of all, it is expressed solely in terms of the three known bounds and a given weight matrix. 
Second, a straightforward optimization for a nontrivial function reads to 
the exactly same expression as the RLD CR bound when the condition 
$C_\theta^R[W]\ge(C_\theta^Z[W]+C_\theta^S[W])/2$ is satisfied. 
As noted before, this condition explicitly depends on the choice of the weight matrix $W$. 
At first sight, this seems to be in contradiction with the general theorem \ref{thmDinv} 
stating that the RLD CR bound can be attained if and only if the model is D-invariant. 
Therefore, we must examine that the Holevo bound 
is identical to the RLD CR bound if and only if the model is D-invariant based on Theorem \ref{thm1}. 
The following proposition confirms that this is indeed so. 
We note that this statement is a special case of Theorem \ref{thmDinv}. 
Here, its proof becomes extremely simple with the obtained formula. 
\begin{proposition}\label{thm10}
For any two-parameter qubit model, the Holevo bound at $\theta$ becomes same as the RLD CR bound 
for all positive weight matrices if and only if the model is D-invariant at $\theta$. That is, 
\be \label{thmeq10}
\forall W \in\cW\ C_\theta^H[W]= C_\theta^R[W]\ \Leftrightarrow\ \mbox{$\cM$ is D-invariant at $\theta$}.
\ee
\end{proposition}
\begin{proof}
When the model is D-invariant, the condition $C_\theta^R[W]<({C_\theta^Z[W]+C_\theta^S[W]})/{2}$ 
cannot be satisfied. (See the remark after Eq.~\eqref{zger}.) Therefore, the Holevo bound 
is always identical to the RLD CR bound in this case. 

Next we show the left condition 
implies the right in Eq.~\eqref{thmeq10}. 
If $C_\theta^H[W]= C_\theta^R[W]$ holds for all $W$, 
the relation $\sqrt{\det W} \left| \vecin{\v \ell_\theta^1}{F_\theta\v \ell_\theta^2}\right|
\ge\mathrm{Tr}\{W(G_\theta^{-1}-\Re\tilde{G}_\theta^{-1} )\}  $ must be true 
for all $W$. 
But, if $G_\theta^{-1}-\Re(\tilde{G}_\theta^{-1})\neq0$, this is impossible. 
To show it, let us suppose $B:=G_\theta^{-1}-\Re\tilde{G}_\theta^{-1} > 0$, 
then we can change $W\to B^{-1/2}W' B^{-1/2}$ to get 
\be \label{proof10}
\frac{\sqrt{\det W'}}{\Tr{W'} }\left| \vecin{\v \ell_\theta^1}{F_\theta\v \ell_\theta^2}\right| \det B^{-1}
\ge 1 ,\ \forall W' \in\cW. 
\ee
It is easy to show that $f(W')=\sqrt{\det W'}/\Tr{W'}$ as a function positive matrix $W'$ satisfies  
$1/2\ge f(W')>0$ as changing positive matrix $W'$. That is $f(W')$ can be made arbitrary small 
by choosing $W'$ and hence the condition \eqref{proof10} cannot hold. This gives 
$B=0\Leftrightarrow G_\theta^{-1}=\Re\tilde{G}_\theta^{-1}$ and hence
this is equivalent to the D-invariant condition from Lemma \ref{lem8}.  $\square$
\end{proof}

Using Theorem \ref{thm1}, we now state one more important condition characterizing 
the model where the Holevo bound coincides with the SLD CR bound. 
\begin{proposition}\label{thm11}
For any two-parameter qubit model, the Holevo bound coincides with the SLD CR bound 
for all positive weight matrices if and only if the imaginary part of $Z_\theta$ vanishes. 
Further, this condition is equivalent to the existence of a weight matrix $W_0$ 
such that the Holevo bound and the SLD CR bound are same with this particular choice $W_0$ \cite{nagaokacomment}. 
Mathematically, we have  
\begin{align} \label{thmeq11} 
&\forall W \in\cW\ C_\theta^H[W]= C_\theta^S[W]\\ \label{thmeq12} 
 \Leftrightarrow&\ \Im Z_\theta=0\\ \label{thmeq13} 
\Leftrightarrow&\ \exists W_0 \in\cW\ C_\theta^H[W_0]= C_\theta^S[W_0]. 
\end{align}
\end{proposition}

\begin{proof}
\eqref{thmeq12} $\Rightarrow$ \eqref{thmeq11}: 
The condition $ \Im Z_\theta=0\Leftrightarrow \Im\rldin{L_{\theta,1}}{L_{\theta,2}}=0$ 
implies $\trAbs\{W\Im \tilde{G}_\theta^{-1}\}=0$ for all $W\in \cW$. 
In this case, the condition $C_\theta^R[W]<({C_\theta^Z[W]+C_\theta^S[W]})/{2}$ 
is always satisfied, and hence the Holevo bound coincides with the SLD CR bound 
for all choices of the weight matrix. 
This follows from the second line of the expression \eqref{sfunction}.

\eqref{thmeq11} $\Rightarrow$ \eqref{thmeq13}: The condition \eqref{thmeq11} implies an existence of a weight matrix satisfying 
$C_\theta^H[W]=C_\theta^S[W]$. 

\eqref{thmeq13} $\Rightarrow$ \eqref{thmeq12}: 
Next let us assume $C_\theta^H[W_0]= C_\theta^S[W_0]$ for some weight matrix $ W_0\in\cW$. 
The expression \eqref{h4bound} immediately implies 
$\Trabs{W_0\Im Z_\theta}=2\sqrt{\det W_0}|\Im\rldin{L_{\theta,1}}{L_{\theta,2}}|=0$
which gives $\Im Z_\theta=0$. Therefore, three conditions are all equivalent. $\square$
\end{proof}

We have three remarks regarding this proposition. 
First, in terms of the SLD Bloch vectors, the necessary and sufficient condition \eqref{thmeq12} is also written as 
\begin{align*}
\Im Z_\theta=0&\Leftrightarrow \tr{\rho_\theta\,[L_{\theta,1},\,L_{\theta,2}]}=0\\
&\Leftrightarrow\vecin{\stheta}{\v \ell_{\theta,1}\times\v \ell_{\theta,2}}=0\\ 
&\Leftrightarrow\  \vecin{\stheta}{ \del_{1}\stheta\times\del_2\stheta}=0,
\end{align*}
which is easy to check by calculating the Bloch vector of a given model. 

Second, we note that given a symmetric matrix $A$, $\Tr{WA}\ge0$ for all positive weight matrix $W$ implies 
$A\ge0$ as a matrix inequality. When the Holevo bound is same as the SLD CR bound, 
we see that the MSE matrix $V^{(n)}_\theta[\hat{\Pi}^{(n)}]$ for any asymptotically unbiased estimators 
satisfy the SLD CR inequality: 
\be
\lim_{n\to\infty}n V^{(n)}_\theta[\hat{\Pi}^{(n)}]\ge G_\theta^{-1},
\ee
as a matrix inequality. 
Moreover, there exists a sequence of estimators that attains this matrix equality. 
This is rather counter intuitive since two SLD operators $L_{\theta,1}$ and $L_{\theta,2}$ 
do not commute in general. Therefore, the condition $\Im\rldin{L_{\theta,1}}{L_{\theta,2}}=0$ 
seems to grasp asymptotic commutativity of two SLD operators in some sense. 
Indeed, this condition can be written as $\tr{\rho_\theta\,[L_{\theta,1},\,L_{\theta,2}]}=0$, 
i.e., commutativity of $L_{\theta,1}$ and $L_{\theta,2}$ on the trace of the state $\rho_\theta$. 
When this holds, the quantum parameter estimation problem becomes 
similar to the classical case asymptotically. 
In the rest of the paper, we call a model {\it asymptotically classical} 
when this condition is satisfied. A similar terminology, ``quasi classical model,'' was 
used by Matsumoto in the discussion of parameter estimation of pure states \cite{matsumoto02}.  
Here, we emphasize that classicality arises only in the asymptotic limit 
and hence, this terminology is more appropriate. 
We also note that the equivalence between \eqref{thmeq12} $\Rightarrow$ \eqref{thmeq13} 
was stated in the footnote of the paper \cite{bbgmm06} based on the unpublished work of Hayashi and Matsumoto. 
Here our proof is shown to be simple owing to the general formula obtained in this paper. 

Last, a great reduction occurs in the structure of the fundamental precision bound for this class of models. 
We note that achievability of the SLD CR bound for specific models 
have been reported in literature in discussion on quantum metrology \cite{cdbw14,vdgjkkdbw14}. 
Here, we provide a simple characterization, the necessary and sufficient condition, 
of such special models in the unified manner. 

Having established the above two propositions, we can conclude that 
a generic two-parameter qubit model other than D-invariant or asymptotically classical ones 
exhibits the nontrivial structure for the Holevo bound in the following sense:  
The structure changes smoothly as the weight matrix $W$ varies. 
For a certain choice of $W$, it coincides with the RLD CR bound and it becomes different 
expression for other choices. 
Put it differently, consider any model that is not asymptotically classical, then we can always 
find a certain region of the weight matrix set $\cW$ such that the Holevo bound is same as the RLD CR bound. 
This point is examined in detail in the next subsection and examples 
are provided in the next section for illustration.  

\subsection{Smooth transition in the Holevo bound}\label{sec4-2}
Let us consider a two-parameter qubit model that is neither D-invariant nor asymptotically classical.  
In this case, the set of all possible weight matrices $\cW=\{W|\mbox{$W$ is $2\times2$ real positive definite}\}$
 is divided into three subsets. The first two sets are $\cW_{+}$ and $\cW_{-}$ in which 
$C_\theta^R[W]-({C_\theta^Z[W]+C_\theta^S[W]})/{2}$ is positive and negative, respectively. 
The last is the boundary $\cW_{\del}$ that consists of a family of weight matrices satisfying the equation: 
\be
B_\theta[W]:=C_\theta^R[W]-\frac12(C_\theta^Z[W]+C_\theta^S[W])=0.
\ee
According to this division, the Holevo bound takes the form $C_\theta^H[W]=C_\theta^R[W]$ 
for $W\in\cW_{+}$ and it is expressed as $C_\theta^H[W]=C^R_\theta[W]+S_\theta[W]$ for $W\in\cW_{-}$. 
Whereas, $C_\theta^H[W]=C_\theta^R[W]=S_\theta[W]$ holds for the boundary $W\in\cW_{\del}$. 

In the following we characterize these sets explicitly. To do this, we first note that 
the degree of freedom for the weight matrix in our problem is three due to 
the condition of $W$ being real symmetric. 
Second, we can show that a scalar multiplication of the weight matrix does not change anything 
except for the multiplication of over all expression of the Holevo bound. 
Thus, we can parametrize the $2\times2$ weight matrix by 
two real parameters. For our purpose, we employ the following 
representation up to an arbitrary multiplication factor: 
\be
W=U_\theta\left(\begin{array}{cc}1 &  w w_2 \\w w_2  & w_2^2\end{array}\right) U_\theta^*, 
\ee
where $w_2>0$ and $\det W=w_2^2(1-w^2)>0\Rightarrow |w|<1$ are imposed from the positivity condition 
and the real orthogonal matrix $U_\theta$ is defined in terms of Eq.~\eqref{gamma} by
\be
U_\theta=\frac{1}{\sqrt{\gamma_{\theta,1}^2+\gamma_{\theta,2}^2}}
\left(\begin{array}{cc}\gamma_{\theta,1}&  -\gamma_{\theta,2} \\ \gamma_{\theta,2}  & \gamma_{\theta,1}\end{array}\right). 
\ee 
The assumption of the model under consideration then yields 
\begin{align}\nonumber
B_\theta[W]&=| \vecin{\v\ell_\theta^1}{F_\theta\v\ell_\theta^2}|\sqrt{\det W}
-\Tr{W(G_\theta^{-1}-\Re\tilde{G}_\theta^{-1})} \\ \nonumber
=& \Big[  |\vecin{\v\ell_\theta^1}{F_\theta\v\ell_\theta^2}|\sqrt{1-w^2}-\Tr{G_\theta^{-1}-\Re\tilde{G}_\theta^{-1}}\,w_2 \Big]w_2. 
\end{align} 
Therefore, by defining a constant solely calculated from the given model 
$\alpha_\theta:= |\vecin{\v\ell_\theta^1}{F_\theta\v\ell_\theta^2}|/\Tr{G_\theta^{-1}-\Re\tilde{G}_\theta^{-1}}$, 
we obtain the sets $\cW_{\pm},\cW_{\del}$ as follows. 
\begin{align}\nonumber
\cW_{+}&=\left\{c\, U_\theta\left(\begin{array}{cc}1 &  \alpha_\theta w w_2 \\ 
\alpha_\theta w w_2  & \alpha_\theta^2 w_2^2\end{array}\right) U_\theta^*
\,\Big|\, w_2^2+w^2<1\right\},\\ \nonumber
\cW_{-}&=\left\{c\, U_\theta\left(\begin{array}{cc}1 &  \alpha_\theta w w_2 \\ 
\alpha_\theta w w_2  & \alpha_\theta^2 w_2^2\end{array}\right) U_\theta^*
\,\Big|\, w_2^2+w^2>1\right\},\\ 
\cW_{\del}&=\left\{c\, U_\theta\left(\begin{array}{cc}1 &  \alpha_\theta w w_2 \\ 
\alpha_\theta w w_2  & \alpha_\theta^2 w_2^2\end{array}\right) U_\theta^*
\,\Big|\, w_2^2+w^2=1\right\},
\end{align}
where the common conditions $w_2>0,|w|<1$, and $c>0$ also need to be imposed 
to satisfy the positivity of the weight matrix. 

\subsection{Limit for pure-state models} \label{sec4-3}
So far we only consider models which consist of mixed states. 
It is known that collective measurements do not improve 
the MSE for pure-state models \cite{GM00,matsumoto02,YFG13}. 
In other words, the Holevo bound 
is same as the bound achieved by separable measurements as 
far as pure-state models are concerned. 

In this subsection, we examine the pure-state limit for our general result. 
When a mixed-state model is asymptotically classical, the Holevo bound 
is identical to the SLD CR bound. This should be true in the pure-state limit, 
and this agrees with the result of Matsumoto \cite{matsumoto02}. 
When a model is D-invariant, on the other hand, it is shown that 
the RLD CR bound can be achieved. This also holds in the pure-state limit 
and we examine the pure-state limit for a generic mixed-state model below.  

We first note that we cannot take the limit for SLD and RLD operators directly \cite{fn95,FN99}. 
This is because there are the terms $1-s_\theta^2$ appearing in the denominators. 
However, the SLD and RLD dual operators are well defined even in the pure-state limit. 
By direct calculation, we can show that the SLD and RLD dual vectors (\ref{sldcot}, \ref{rldcot}) are written as 
\begin{align*}
\v{\ell}_\theta^1&=-\frac{Q_\theta^{-1}\v{\ell}_\theta^\bot\times\del_2\stheta}{\vecin{\v{\ell}_\theta^\bot}{Q_\theta^{-1}\v{\ell}_\theta^\bot}},\\ 
\v{\ell}_\theta^2&=\frac{Q_\theta^{-1}\v{\ell}_\theta^\bot\times\del_1\stheta}{\vecin{\v{\ell}_\theta^\bot}{Q_\theta^{-1}\v{\ell}_\theta^\bot}},\\
\tilde{\v{\ell}}_\theta^1&=-\frac{(\openone-\I F_\theta)(\v{\ell}_\theta^\bot\times\del_2\stheta-\I\vecin{\stheta}{\v{\ell}_\theta^\bot}\del_2\stheta)}
{\vecin{\v{\ell}_\theta^\bot}{Q_\theta^{-1}\v{\ell}_\theta^\bot}},\\
\tilde{\v{\ell}}_\theta^2&=\frac{(\openone-\I F_\theta)(\v{\ell}_\theta^\bot\times\del_1\stheta-\I\vecin{\stheta}{\v{\ell}_\theta^\bot}\del_1\stheta)}
{\vecin{\v{\ell}_\theta^\bot}{Q_\theta^{-1}\v{\ell}_\theta^\bot}},
\end{align*}
for any two-parameter qubit model. 
Thus, as long as $\vecin{\v{\ell}_\theta^\bot}{Q_\theta^{-1}\v{\ell}_\theta^\bot}\neq0$, 
they converge in the pure-state limit. This condition is also expressed as
$|\v{\ell}_\theta^\bot\times\stheta|^2-(1-s_\theta^2)^2|\v{\ell}_\theta^\bot|^2\neq0$, 
and hence is equivalent to $\v{\ell}_\theta^\bot\times\stheta\neq\v{0}$ in the pure-state limit., 
i.e., $\cM$ is not asymptotically classical. 
Second, the same warning applies to the SLD and RLD Fisher information \cite{fn95,FN99}. 
However, the inverse of the SLD Fisher information matrix is well defined 
even for the pure-state limit. This is because the $(i,j)$ component of the inverse of SLD Fisher information 
matrix is $\vecin{\v{\ell}_\theta^i}{Q_\theta^{-1}\v{\ell}_\theta^j}$ and this has the well-defined limit. 
The same reasoning can be applied to the inverse of the RLD Fisher information matrix in the pure-state limit. 

Last, let us examine the general formula. It is straightforward 
to show that the function \eqref{hsbound} vanishes in the pure-state limit.  
In other words, the general formula given in Theorem \ref{thm1} becomes 
$C^H_\theta[W]=C^R_\theta[W]$ for all weight matrices. 
Therefore, the Holevo bound becomes the RLD CR bound in the pure-state limit. 
This is, of course, expected because any pure-state model can be expressed 
as a unitary model, which is D-invariant from Lemma \ref{lem8}. 

\section{Examples}\label{sec5}
In this section, we consider several examples for two-parameter qubit models to 
illustrate our result. The first one is a D-invariant model whose Holevo bound is 
identical to the RLD CR bound. The second one is a asymptotically classical model 
which gives the SLD CR bound. As the last example, we analyze a generic model 
in particular the behavior of the Holevo bound depending on the weight matrix. 

Within the setting of pointwise estimation, 
we note that it is sufficient to specify the following three vectors to define a model locally;
\be\label{localmodel}
\{\stheta,\, \del_1\stheta,\,\del_2\stheta\} \quad\mathrm{at}\,\theta.   
\ee
Equivalently, the set $\{\stheta,\v\ell_{\theta,1},\v\ell_{\theta,2}\}$ or $\{\stheta,\tilde{\v\ell}_{\theta,1},\tilde{\v\ell}_{\theta,2}\}$ 
can be used to define the model at $\theta$ uniquely. 
One-to-one correspondences among these three specifications of the model are easily established 
by the Bloch vector representation discussed in Sec.~\ref{sec3-1}. 

The D-invariant condition is now expressed in terms of the vectors of the set \eqref{localmodel} as
\be
\mbox{D-invariant}\ \Leftrightarrow\ \vecin{\stheta}{\del_1\stheta}=\vecin{\stheta}{\del_2\stheta}=0. 
\ee
The asymptotically classical condition is equivalent to 
\be
\mbox{Asymptotically classical}\  \Leftrightarrow\ \vecin{\stheta}{\del_1\stheta \times \del_2\stheta}=0 . 
\ee
All models other than satisfying the above two conditions are generic ones. 

\subsection{D-invariant model}\label{sec5-1}
Consider a simple unitary model where 
the parametric state is generated by a two-parameter unitary $U(\theta)$ ($\theta=(\theta_1,\theta_2)$): 
\be
\rho_\theta= U(\theta)\rho_0U(\theta)^*. 
\ee
It is easy to show that the norm of the Bloch vector is independent of $\theta$,  
and hence, this model is globally D-invariant.  
The Holevo bound is $C_\theta^H[W]=C_\theta^R[W]$ for all weight matrices $W$.

We note that this kind of D-invariant models possesses symmetry and 
has been studied by many authors, see for example Ch.~4 of Holevo's book \cite{holevo} 
and Hayashi \cite{hayashi} and references therein. 

\subsection{Asymptotically classical model}\label{sec5-2}
Consider the following model: 
\be
\cM_{\cB}=\left\{\stheta=f^1(\theta)\v u_1+f^2(\theta)\v u_2\, |\, \theta\in\Theta \right\},
\ee
where $\v u_i$ are unit vectors, which are not necessarily orthogonal to each other, 
and $f^j(\theta)$ are scalar (differentialble) functions of $\theta=(\theta^1,\theta^2)$. 
The parameter region $\Theta$ is specified by an arbitrary open subset 
of the set; $\{(\theta_1,\theta_2)\,|\, s_\theta<1  \}$.
We can show that this model is asymptotically classical because of 
$\stheta\bot \del_1\stheta \times \del_2\stheta\propto \v u_1\times\v u_2$, 
and the Holevo bound is $C_\theta^H[W]=\Tr{WG_\theta^{-1}}\forall W>0$. 

For the asymptotically classical case, we can easily compare this result with 
the bound achieved by an optimal estimator comprised of separable POVMs. 
In this case, the Nagaoka bound is known to be achievable that is 
calculated as \cite{nagaoka89} 
\be
C_\theta^N[W]=\Tr{WG_\theta^{-1}}+2\sqrt{\det WG_\theta^{-1}}. 
\ee
Therefore, the gain due to use of collective POVMs is $2\sqrt{\det WG_\theta^{-1}}$. 

As a special case, let us set $\v u_j$ to be orthogonal normal vectors 
and $f^j(\theta)=\theta^j$. Then, the inverse of the SLD Fisher information matrix reads
\be
G_\theta^{-1}=\left(\begin{array}{cc}1-(\theta^1)^2  & -\theta^1\theta^2 \\ -\theta^1\theta^2  & 1-(\theta^2)^2\end{array}\right), 
\ee 
and the gain mentioned above is $2\sqrt{\det W}\sqrt{1-s_\theta^2}$. 
Thus, the role of collective POVMS becomes important as the state becomes more mixed. 

\subsection{Generic model}\label{sec5-3}
In this subsection, we analyze a generic model other than the previous two examples. 
In this case the structure of the Holevo bound changes when the weight matrix $W$ varies, 
that is it takes the same form as the RLD CR bound for certain choices of $W$ whereas 
it becomes different form for other cases. This is explicitly demonstrated in the section \ref{sec4-2} 
and here we consider a simple yet nontrivial example to gain a deeper insight into this a phenomenon of 
this transition. 

Consider a two-parameter qubit model: 
\be\label{gicmodel}
\cM_{\cB}=\left\{\stheta=(\theta^1,\theta^2,\theta_0)^T\, |\, \theta\in\Theta \right\},
\ee
where $\theta_0$ is a fixed parameter which satisfies $0<|\theta_0|<1$. 
The parameter region is specified by a subset of 
$\{(\theta^1,\theta^2)|(\theta^1)^2+(\theta^2)^2<1-\theta_0^2 \}$. 
This model is neither D-invariant nor asymptotically classical 
when $\theta_0\neq0$ and $\theta^1\theta^2\neq0$ are satisfied. 
This model becomes asymptotically classical for $\theta_0=0$ as discussed before. 
We can also regard this model as a sub-model embedded in the three-parameter qubit model: 
$\left\{\stheta=(\theta^1,\theta^2,\theta^3)^T\, |\, \theta\in\Theta\subset\bbr^3 \right\}$. 

The inverse of SLD and RLD Fisher information matrices at $\theta$ are
\begin{align}\nonumber
G_\theta^{-1}&=\frac{1}{1-\theta_0^2}\left(\begin{array}{cc}1-\theta_0^2-(\theta^1)^2 & -\theta^1\theta^2 \\
-\theta^1\theta^2 & 1-\theta_0^2-(\theta^2)^2\end{array}\right) ,\\ 
\tilde{G}_\theta^{-1}&=\frac{1-s_\theta^2}{1-\theta_0^2}
\left(\begin{array}{cc}1 & -\I\theta_0 \\ \I\theta_0 & 1\end{array}\right)
\end{align}
with $s_\theta=[(\theta^1)^2+(\theta^2)^2+\theta_0^2]^{1/2}$, respectively, 
and the three bounds appearing in Theorem \ref{thm1} read
\begin{align}\nonumber
C_\theta^S[W]&=\Tr{W}- \frac{1}{1-\theta_0^2} (\vec{\theta}|W\vec{\theta}),\\ \nonumber
C_\theta^R[W]&= \frac{1-s_\theta^2}{1-\theta_0^2} \Tr{W}
+2\sqrt{\frac{1-s_\theta^2}{1-\theta_0^2}}\, |\theta_0| \sqrt{\det W},\\ \nonumber
C_\theta^Z[W]&= \Tr{W}- \frac{1}{1-\theta_0^2} (\vec{\theta}|W\vec{\theta})+2\sqrt{\frac{1-s_\theta^2}{1-\theta_0^2}}\, |\theta_0| \sqrt{\det W},  
\end{align}
where $|\vec\theta)=(\theta^1,\theta^2)^T$ is introduced for convenience. 

In the following let us analyze the structure of the Holevo bound using 
the following parametrization of a weight matrix: 
\begin{multline}\label{weightex}
W(w,\omega)=R(\omega)
\left(\begin{array}{cc}\frac12({1+w})& 0 \\0 & \frac12({1-w})\end{array}\right)
R(\omega)^T,\ \\
 \mathrm{with}\quad R(\omega)=
\left(\begin{array}{cc}\cos\omega & -\sin\omega \\ \sin\omega & \cos\omega\end{array}\right), 
\end{multline}
where we normalize the trace of $W$ to be one and $w\in(-1,1)$, $\omega\in[0,2\pi)$. 
This parametrization is different from the one analyzed in Sec.~\ref{sec4-2}, 
yet is convenient for the purpose of visualization.  
It is easy to see that the effect of the matrix $R(\omega)$ is to mix two parameters $\theta^1$ and $\theta^2$ 
by rotating them about an angle $\omega$. 
Since $\det W=(1-w^2)/4$ for this particular parametrization, 
we see that the RLD CR bound is independent of the weight parameter $\omega$. 
The other bounds depend on two parameters $(w,\omega)$. 

We are interested in how the Holevo bound $C^H_\theta[w,\omega]$ changes when we vary the weight parameters $w,\omega$. 
Upon plotting, we fix the direction of the estimation parameter $\vec\theta$ as $\vec\theta=|\theta| (1,1)/\sqrt{2}$ 
with $|\theta|=(\theta|\theta)^{1/2} $. 
We plot the Holevo bound for two sets of the state parameters; (a) $\theta_0=0.2, \vec\theta=0.346 (1,1)/\sqrt2$ 
and (b) $\theta_0=0.275, \vec\theta=0.476 (1,1)/\sqrt2$ for illustration. 
Figure \ref{fig1a} shows the Holevo bound for the state parameter (a) as a function of the weight-matrix parameter $w,\omega$. 
In this plot, the gray areas indicate the case for $C^H_\theta[W]=C^R_\theta[W]+S_\theta[W]$, 
whereas the white-meshed region indicates the case for $C^H_\theta[W]=C^R_\theta[W]$. 
We also show the other state parameter setting (b) in Fig.~\ref{fig1b} with the same convention. 
In both figures we observe smooth transitions between two-different expressions discussed in Sec.~\ref{sec4-2}.  
\begin{figure}[htbp]
\subfigure[]{
\label{fig1a}
\includegraphics[width=5cm]{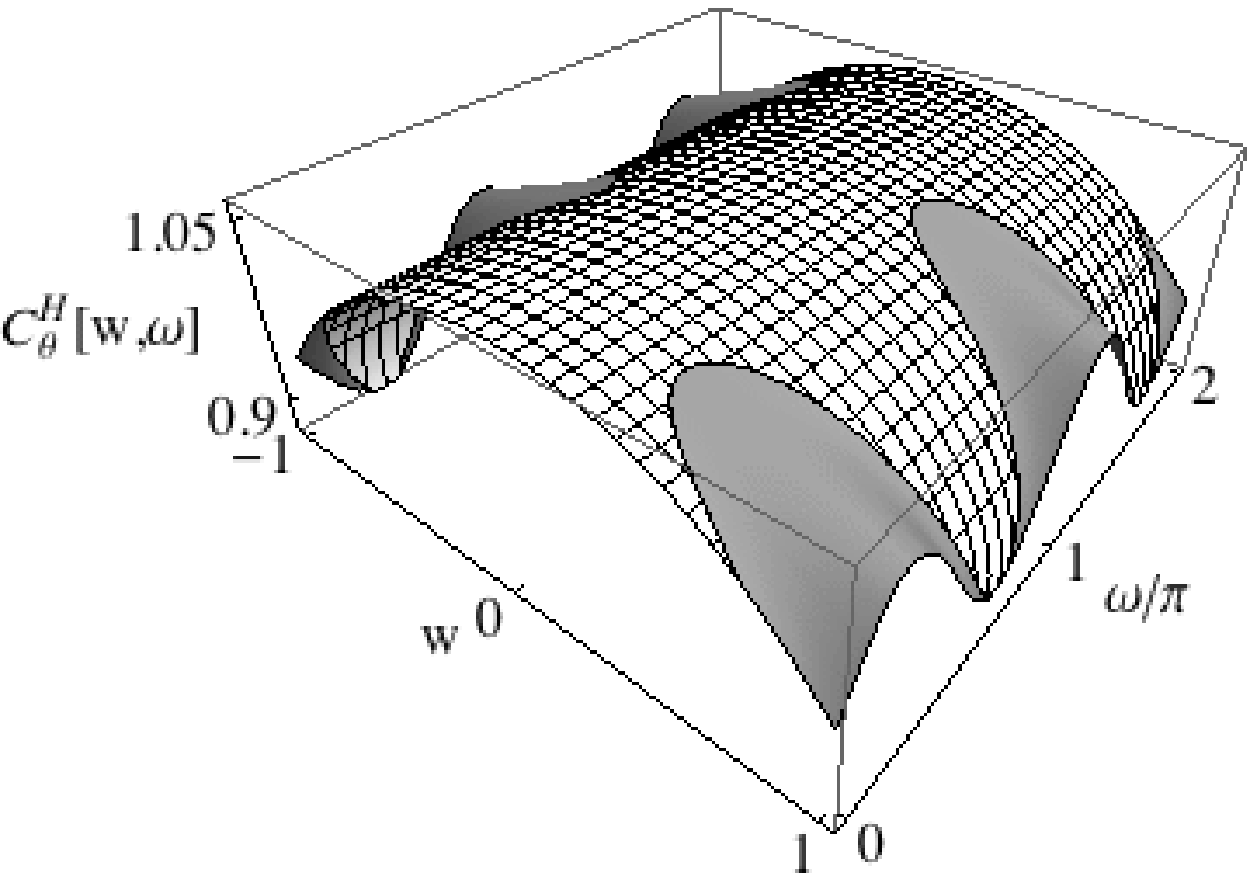}}
\hspace{1mm}
\subfigure[]{
\label{fig1b}
\includegraphics[width=5cm]{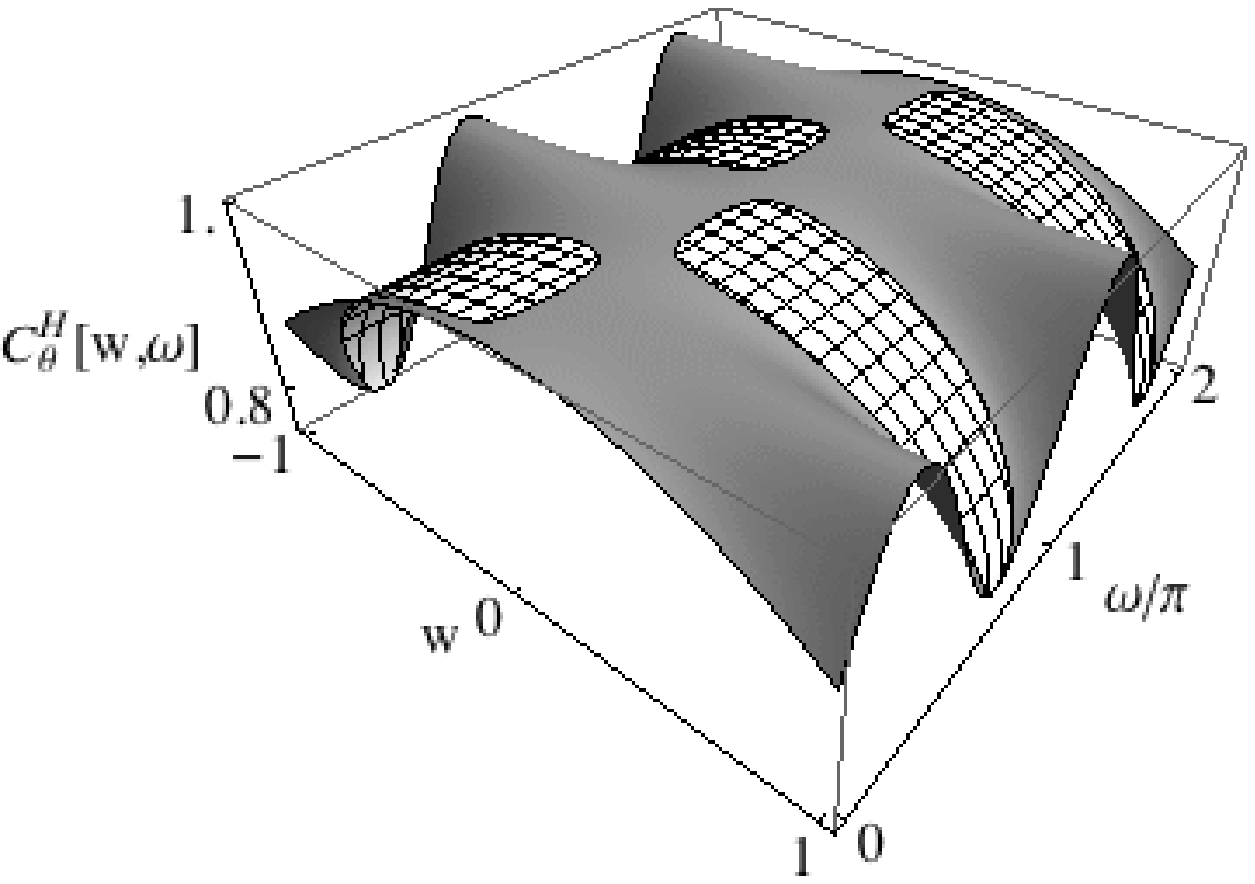}}
\caption{The Holevo bound for the state parameter (a) $\theta_0=0.2, \vec\theta=0.346 (1,1)/\sqrt2$ 
and (b) $\theta_0=0.275, \vec\theta=0.476 (1,1)/\sqrt2$ as a function of 
the wight-matrix parameter $(w,\omega)$ given in Eq.~\eqref{weightex}. 
The gray areas indicate the case for $C^H_\theta[W]=C^R_\theta[W]+S_\theta[W]$, 
whereas the white-meshed region indicate the case for $C^H_\theta[W]=C^R_\theta[W]$.}
\end{figure}

From these figures, we see that the Holevo bound coincides with the RLD CR bound 
for relatively large choice of the weight matrix in Fig.~\ref{fig1a}, 
where as the opposite observation holds for Fig.~\ref{fig1b}.  
To gain a deeper insight, let us calculate the value of quantities 
$\gamma_{\theta,i}:=\vecin{\stheta}{\v{\ell}_{\theta,i}}$ ($i=1,2$). 
Then, we get $\gamma_{\theta,1}=\gamma_{\theta,2}=0.292$ for the case (a) 
and $\gamma_{\theta,1}=\gamma_{\theta,2}=0.483$ for the case (b).  
Since vanishing of these quantities is equivalent to the D-invariance of a model (Lemma \ref{lem8}), 
we naively expect that the smaller values of them imply that a model behaves more D-invariant-like. 
Indeed, the examples presented here agree with our intuition, yet more detail analysis is 
needed to make any conclusion. 

\subsection{Transition in the parameter space}\label{sec5-4}
We briefly discuss another important observation of this paper. 
A generic model, other than special cases discussed before, exhibits a transition in the structure of the Holevo bound 
for a fixed weight matrix when we change the estimation parameter $\theta$. 
A rough sketch of this argument is that 
a change in the weight matrix is amount to that in the parameter 
and vice versa. This is a well-known fact in the pointwise estimation setting \cite{nagaoka91,FN99}. 
Below we briefly show such an example. The model is same 
as the generic model \eqref{gicmodel} analyzed the previous subsection. 

As noted before, a change in the weight-matrix parameter $\omega_0\to\omega_0+\omega$ is equivalent 
to rotate the parameter $(\theta^1,\theta^2)$ by the angle $\omega$. 
Depending upon the choice of the weight matrix, we can also observe 
a similar transition when we change the parameter $\theta$ in the set $\Theta$. 
We set the weight matrix to be a diagonal matrix $W_0=\mathrm{diag} (0.55,0.45)$ 
and $\theta_0=0.35$. Figure \ref{fig2} plots the Holevo bound 
as a function of the state parameter $\theta=(\theta^1,\theta^2)$. 
The gray area indicates the region where $C^H_\theta[W_0]=C^R_\theta[W_0]+S_\theta[W_0]$ holds, 
whereas the white-meshed region indicates the case for $C^H_\theta[W_0]=C^R_\theta[W_0]$. 
This shows that the Holevo bound coincides with the RLD CR bound for 
a certain subset of the parameter space. 
\begin{figure}[htbp]
\begin{center}
\includegraphics[width=5.5cm,keepaspectratio,clip]{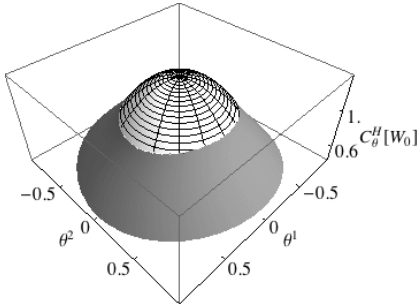}
\caption{The Holevo bound for the weight matrix $W_0=\mathrm{diag} (0.55,0.45)$ and $\theta_0=0.35$ 
as a function of the state parameter $\vec\theta=(\theta^1,\theta^2)$. 
The gray area indicates the case for $C^H_\theta[W]=C^R_\theta[W]+S_\theta[W]$, 
whereas the white-meshed region indicates the case for $C^H_\theta[W]=C^R_\theta[W]$.}
\label{fig2}
\end{center}
\end{figure}

\section{Concluding remarks}\label{sec6}
In this paper, we have derived a closed formula for the fundamental precision bound, the Holevo bound, 
for any two-parameter qubit estimation problem in the pointwise estimation setting. 
This bound is known to be asymptotically achievable by the optimal sequence of estimators 
consisting of jointly performed measurements under some regularity conditions. 
Since there exist explicit formulas for the Holevo bound for the pure-state qubit model, qubit mixed-state models with 
one and three parameters, our result completes a list of the fundamental precision bounds 
in terms of quantum versions of Fisher information, which is calculated from a given quantum parametric model, 
as far as qubit models are concerned.  
The obtained formula shows several new insights into the property of 
the Holevo bound for quantum parameter estimation problems. 
In the following we shall list concluding remarks together with outlook for future works. 

First, the necessary and sufficient conditions for the asymptotic achievability of 
the SLD and RLD CR bounds are derived when estimating any two-parameter family of qubit states. 
In particular, the notation of asymptotically classical models is proposed, 
in which all SLD operators commute with each other on the trace of a given parametric state. 
In this case, the weight matrix can be eliminated and the problem becomes similar to the classical statistics in the asymptotic limit, 
where the SLD Fisher information plays the same role as the Fisher information. 
We note that the notion of asymptotic classicality can be 
extended to any models on any finite-dimensional system and the same statement obtained for the qubit case holds: 
The Holevo bound coincides with the SLD CR bound if and only if the model is asymptotically 
classical. The detail of this result shall be given in the subsequent paper.  

Second, the RLD CR bound is shown to be achieved for a certain choice of the weight matrix 
even though the model is not D-invariant. This result emphasizes the importance of 
the RLD Fisher information for general qubit parameter estimation problems. 
Since the imaginary parts of the inverse of the RLD Fisher information matrix 
and the $Z_\theta$ matrix are always identical for two-parameter qubit mixed-state models, 
we cannot immediately conclude that this is the general statement or not. It might well be that 
the only real part of the inverse of RLD Fisher information matrix plays an important role 
in higher dimensional systems. This deserves further studies and shall be analyzed as 
an extension of the present work. 

Third, our result also provides the (unique) minimizer to the optimization problem appearing in 
the definition of the Holevo bound. This set of observables, which are locally unbiased in the sense of Eq.~\eqref{lucondition}, 
can be used to construct an optimal sequence of POVMs that attains the Holevo bound asymptotically. 
This line of approach has been proposed by Hayashi, who reported a theorem without proof 
to realize the construction of POVMs on the tensor product states \cite{hayashi03}. 
His approach is different from other approaches given in Refs.~\cite{HM08,GK06,YFG13} 
and it may need more refined arguments to complete his theorem. 

Last, smooth transitions in the structure of the Holevo bound is 
shown to occur in general. Since this happens in the simplest quantum system, 
we expect similar phenomena occur in higher dimensional systems as well.   
However, we do not know whether or not the number of different forms is always two as 
demonstrated here. The techniques used in this paper can be 
applied to two-parameter estimation problems in higher dimensional systems, 
and we shall make progress in due course. 

\section*{Acknowledgement}
The author is indebted to Prof.~H.~Nagaoka for invaluable discussions and suggestions 
to improve the manuscript. 
In particular, Theorem \ref{thmDinv} and Lemma \ref{lem4} were motivated by his seminar \cite{nagaokaseminar}. 
He also thanks B.-G. Englert and C.~Miniatura for their kind hospitality at the 
Centre for Quantum Technologies, Singapore, where a part of this work was done. 

\appendix

\section{Proofs for lemmas}\label{sec:appb}
\subsection{Leema \ref{lem-holevo2}}\label{sec:appb-0}
Proof follows from a straightforward calculation. 
We substitute $\v x^i=\v\ell_\theta^i+\xi^i \v\ell_{\theta}^{\bot}$ 
into the second line of Eqs.~\eqref{h2function}.  
The first term reads
\begin{align}\nonumber
\vecin{\v x^i}{Q_\theta^{-1}\v x^j}&=
\vecin{\v{\ell}_\theta^i}{Q_\theta^{-1}\v{\ell}_\theta^j}+\xi^i\xi^j \vecin{\v\ell_{\theta}^\bot}{{Q}_{\theta}^{-1}\v\ell_{\theta}^\bot}\\
&=g_\theta^{ij}+ \xi^i\xi^j \vecin{\v\ell_{\theta}^\bot}{{Q}_{\theta}^{-1}\v\ell_{\theta}^\bot},
\end{align}
where the relation $\vecin{\v\ell_{\theta}^i}{Q_\theta^{-1}\v\ell_{\theta}^\bot}=0$ for $i=1,2$ 
and Eq.~\eqref{sldQinv} are used. 
Note $\vecin{\v a}{F_\theta \v a}=0$ for all $\v a\in\bbr^2$, then the second term is calculated as 
\begin{align}
\sum_{i,j=1}^2 |\vecin{\v x^i}{F_\theta\v x^j}|&= 2|\vecin{\v x^1}{F_\theta\v x^2}|\\ \nonumber
=2|\vecin{\v{\ell}_\theta^1}{F_\theta\v{\ell}_\theta^2}&
+\xi^1 \vecin{\v{\ell}_\theta^\bot}{F_\theta\v{\ell}_\theta^2}
+\xi^2 \vecin{\v{\ell}_\theta^1}{F_\theta\v{\ell}_\theta^\bot}|\\ \nonumber
=2|\vecin{\v{\ell}_\theta^1}{F_\theta\v{\ell}_\theta^2}&
+\xi^1(1-s_\theta^2) \gamma_{\theta,1}
+\xi^2(1-s_\theta^2) \gamma_{\theta,2}|, 
\end{align}
where $\vecin{\v{\ell}_\theta^\bot}{F_\theta\v{\ell}_\theta^2}=(1-s_\theta^2) \gamma_{\theta,1}$ 
and $\vecin{\v{\ell}_\theta^\bot}{F_\theta\v{\ell}_\theta^1}=-(1-s_\theta^2) \gamma_{\theta,2}$ 
are used. 
Combining the above expressions, we get the expression of this lemma. $\square$

\subsection{Lemma \ref{lem7}}\label{sec:appb-2}
\noindent
1. $\vecin{\v\ell_{\theta}^\bot}{{Q}_{\theta}^{-1}\v\ell_{\theta}^\bot} 
= (1-s_\theta^2)\det\, G_\theta=(1-s_\theta^2)^2\det \tilde{G}_\theta$:\\
For convenience, let us define 
\begin{align}
\v{n}_\theta&:=\del_1\stheta\times\del_2\stheta,\\ \nonumber
\v{m}_\theta&:=\v\ell_{\theta,1}\times\v\ell_{\theta,2}, 
\end{align}
which are related by $\v{n}_{\theta}=(1-s_\theta^2) Q_\theta \v{m}_\theta$. 
The standard vector analysis shows 
\begin{align*}
\vecin{\v{m}_\theta}{\v{m}_\theta}&=
\vecin{\v\ell_{\theta,1}}{\v\ell_{\theta,1}}\vecin{\v\ell_{\theta,2}}{\v\ell_{\theta,2}}-(\vecin{\v\ell_{\theta,1}}{\v\ell_{\theta,2}})^2,\\
\stheta\times\v{m}_\theta&=\gamma_{\theta,2}\v\ell_{\theta,1}-\gamma_{\theta,1}\v\ell_{\theta,2}. 
\end{align*}

From the expression \eqref{sldqubit}, the components of the SLD Fisher information matrix 
is expressed as $g_{\theta,ij}=\vecin{\v\ell_{\theta,i}}{\v\ell_{\theta,j}}-\gamma_{\theta,i}\gamma_{\theta,j}$. 
Then, the determinant of the SLD Fisher information matrix is calculated as
\begin{align*}
\det \,G_\theta&=(\vecin{\v\ell_{\theta,1}}{\v\ell_{\theta,1}}-\gamma_{\theta,1}^2 )
(\vecin{\v\ell_{\theta,2}}{\v\ell_{\theta,2}}-\gamma_{\theta,2}^2 )\\
&\quad-(\vecin{\v\ell_{\theta,1}}{\v\ell_{\theta,2}}-\gamma_{\theta,1}\gamma_{\theta,2} )^2\\
&=\vecin{\v\ell_{\theta,1}}{\v\ell_{\theta,1}}\vecin{\v\ell_{\theta,2}}{\v\ell_{\theta,2}}-(\vecin{\v\ell_{\theta,1}}{\v\ell_{\theta,2}})^2\\
&\quad-|\gamma_{\theta,2}\v\ell_{\theta,1}-\gamma_{\theta,1}\v\ell_{\theta,2}|^2\\
&=\vecin{\v{m}_\theta}{\v{m}_\theta}-\vecin{\stheta\times\v{m}_\theta}{\stheta\times\v{m}_\theta}\\
&=(1-s_\theta^2)\vecin{\v{m}_\theta}{Q_\theta\v{m}_\theta}\\
&=\frac{1}{1-s_\theta^2}\vecin{\v{n}_\theta}{{Q}_\theta^{-1} \v{n}_\theta},
\end{align*}
and hence we obtain $\vecin{\v\ell_{\theta}^\bot}{{Q}_{\theta}^{-1}\v\ell_{\theta}^\bot} = (1-s_\theta^2)\det\, G_\theta$. 

Next, we use the relation for the RLD Fisher information; 
$\tilde{g}_{\theta,ij}=(1-s_\theta^2)^{-1}[\vecin{\v\ell_{\theta,i}}{\v\ell_{\theta,j}}-(2-s_\theta^2)\gamma_{\theta,i}\gamma_{\theta,j}
-\I\vecin{\v\ell_{\theta,i}}{F_\theta\v\ell_{\theta,j}}]$ to calculate the determinant 
$\det \,\tilde{G}_\theta$ as
\begin{align*}
&(1-s_\theta^2)^2\det \,\tilde{G}_\theta\\
&=(\vecin{\v\ell_{\theta,1}}{\v\ell_{\theta,1}}-(2-s_\theta^2)\gamma_{\theta,1}^2 )
(\vecin{\v\ell_{\theta,2}}{\v\ell_{\theta,2}}-(2-s_\theta^2)\gamma_{\theta,2}^2 )\\
&\quad-|\vecin{\v\ell_{\theta,1}}{\v\ell_{\theta,2}}
-(2-s_\theta^2)\gamma_{\theta,1}\gamma_{\theta,2}
-\I \vecin{\v\ell_{\theta,1}}{F_\theta\v\ell_{\theta,2}}|^2\\
&=\vecin{\v{m}_\theta}{\v{m}_\theta}-(2-s_\theta^2)\vecin{\stheta\times\v{m}_\theta}{\stheta\times\v{m}_\theta}
-|\vecin{\stheta}{\v{m}_{\theta}}|^2\\
&=(1-s_\theta^2)^2\vecin{\v{m}_\theta}{Q_\theta\v{m}_\theta}. 
\end{align*}
This proves the relation; $(1-s_\theta^2)\det \,\tilde{G}_\theta=\det\,G_\theta$. $\square$ 

\noindent
2. $2\sqrt{\det W}\left| \vecin{\v \ell_\theta^1}{F_\theta\v \ell_\theta^2}\right|
=\mathrm{TrAbs}\{W\Im\tilde{G}_\theta^{-1}\}=\mathrm{TrAbs}\{W\Im Z_\theta\}$:\\
From the definition for the matrix $Z_\theta$, the imaginary part is expressed as
\be
\Im Z_\theta=\vecin{\v\ell_{\theta}^1}{F_\theta\v\ell_{\theta}^2}\left(\begin{array}{cc}0 & 1 \\-1 & 0\end{array}\right),  
\ee
and the straightforward calculation yields
\be
\Trabs{W\Im Z_\theta}=2\sqrt{\det W}\left| \vecin{\v \ell_\theta^1}{F_\theta\v \ell_\theta^2}\right|. 
\ee

The imaginary part of the RLD Fisher information matrix is 
\be
\Im \tilde{G}_\theta=\frac{ \vecin{\v\ell_{\theta,1}}{F_\theta\v\ell_{\theta,2}} }{1-s_\theta^2}\left(\begin{array}{cc}0 & -1 \\1 & 0\end{array}\right),
\ee
and the imaginary part of the inverse is 
\be
\Im \{\tilde{G}_\theta^{-1}\}=\frac{ \vecin{\v\ell_{\theta,1}}{F_\theta\v\ell_{\theta,2}} }{\det {G}_\theta}
\left(\begin{array}{cc}0 & 1 \\-1 & 0\end{array}\right),
\ee
where we use $(1-s_\theta^2)\det \,\tilde{G}_\theta=\det\,G_\theta$. 
It is easy to show 
$\vecin{\v \ell_\theta^1}{F_\theta\v \ell_\theta^2}=\vecin{\v\ell_{\theta,1}}{F_\theta\v\ell_{\theta,2}} /\det G_\theta$ 
and thus we obtain the important relationship; $\Im Z_\theta=\Im \tilde{G}_\theta^{-1}$. 
This proves the claim. $\square$

\noindent
3. $(\vec{\gamma}_\theta|W^{-1}\vec{\gamma}_\theta)=
{\det W^{-1}G_\theta}\left(C_\theta^Z[W]-C_\theta^R[W] \right)/{(1-s_\theta^2)}$: \\
This can be shown by the following calculations: 
\begin{align*}
&\det W\times(\vec{\gamma}_\theta|W^{-1}\vec{\gamma}_\theta)\\
&=
w_{22} \gamma_{\theta,1}^2+w_{11} \gamma_{\theta,2}^2-2 w_{12} \gamma_{\theta,1}\gamma_{\theta,2}\\
&=(1-s_\theta^2)^{-1}(w_{22}g_{\theta,11}+w_{11}g_{\theta,22}-2w_{12}g_{12})\\
&\quad-\Re (w_{22}\tilde{g}_{\theta,11}+w_{11}\tilde{g}_{\theta,22}-2w_{12}\tilde{g}_{12})\\
&=(1-s_\theta^2)^{-1}\det G_\theta \Tr{WG_\theta^{-1}}-\det\tilde{G}_\theta \Tr{W\Re \tilde{G}_\theta^{-1}}\\
&=\frac{\det G_\theta }{1-s_\theta^2} \Tr{W(G_\theta^{-1}-\Re \tilde{G}_\theta^{-1})}\\
&=\frac{\det G_\theta }{1-s_\theta^2}\left(C_\theta^Z[W]-C_\theta^R[W] \right). \ \square
\end{align*}

\subsection{Lemma \ref{lem8}}\label{sec:appb-3}
Since the imaginary parts of the inverse of the RLD Fisher information matrix 
and the matrix $Z_\theta$ are always identical for two-parameter qubit models, 
equivalence between 1 and 2 is the direct consequence of Lemma \ref{lem4}. 
We next note the following relation:
\[
G_\theta^{-1}-\Re\tilde{G}_\theta^{-1}=\frac{1}{\det \tilde{G}_\theta}
\left(\begin{array}{cc}   \gamma_{\theta,2}^2&- \gamma_{\theta,1} \gamma_{\theta,2} \\
 - \gamma_{\theta,1} \gamma_{\theta,2} &  \gamma_{\theta,1}^2\end{array}\right), 
\]
which follows from Eq.~\eqref{sldQrld} and Lemma \ref{lem7}-1. 
Thus $G_\theta^{-1}-\Re\tilde{G}_\theta^{-1}=0$ if and only if 
$\gamma_{\theta,1}=\gamma_{\theta,2}=0$. 

To show the statement about global D-invariance, 
we only need to show $\gamma_{\theta,1}=\gamma_{\theta,2}=0$ 
for all $\theta\in\Theta$ if and only if $s_\theta=|\stheta|$ is independent of $\theta$. 
When $\gamma_{\theta,i}= \vecin{\stheta}{\v\ell_{\theta,i}}=0$ 
$\Leftrightarrow$ $\vecin{\stheta}{\del_i\stheta}=(1/2)\del_i |\stheta|^2 =0$ holds, 
the integration of this condition gives $|\stheta|$ is independent of $\theta=(\theta^1,\theta^2)$. 
Conversely, $|\stheta|$ does not depend on $\theta$, we obtain $\gamma_{\theta,i}=0$ ($i=1,2$) 
for all $\theta\in\Theta$. $\square$

\subsection{Lemma \ref{lem9}}\label{sec:appb-4}
When $\vec{b}=\vec{0}$, the function to be minimized is 
$f(\vec{\xi})=(\vec{\xi}|A\vec{\xi})+2|c|$. Since $A$ is positive-definite, 
the minimum is $2|c|$ and is attained by $\vec{\xi}_*=\vec{0}$. 

For the other case ($\vec{b}\neq\vec{0}$), introducing the new variables through 
\[
\vec\eta=BA^{1/2}\vec\xi,\ \mathrm{with}\ 
B= \left(\begin{array}{c}\vec{b}^TA^{-1/2}\\ \vec{b}^TA^{-1/2} \left(\begin{array}{cc}0 & -1 \\1 & 0\end{array}\right)\end{array}\right), 
\]
we can express the function $f(\vec{\xi})$ as
\be
f(\vec{\eta})=\frac{1}{\alpha}(\eta_1^2+\eta_2^2)+2|\eta_1+c|, 
\ee
where $\alpha=(\vec{b}|A^{-1}\vec{b})$ is a positive constant.  
The minimum of this simple quadratic function is obtained by analyzing 
the case $\eta_1+c>0$ and $\eta_1+c\le0$ separately. 
The result is 
\be
\min_{\vec{\eta}\in\bbr^2}f(\vec{\eta})=
\begin{cases}
-\alpha+2c\quad\mathrm{if}c\ge\alpha\\
\ds\frac{c^2}{\alpha}\qquad\qquad\mathrm{if} |c|<\alpha\\
-\alpha-2c\quad\mathrm{if}c\le-\alpha
\end{cases},
\ee
and the unique minimizer $\vec{\eta_*}=(\zeta_*,0)^T$ is
\be
\zeta_*=
\begin{cases}
-\alpha\quad\mathrm{if}c\ge\alpha\\
-c\quad\quad\mathrm{if} |c|<\alpha\\
\alpha\qquad\mathrm{if}c\le-\alpha
\end{cases}.
\ee
This solution can be translated into the original variables $\vec{\xi}$ 
to prove the lemma. $\square$

\section{Supplemental materials}\label{sec:appc}
\subsection{Canonical projection and invariant subspace}\label{sec:app-3}
In this appendix, we summarize the concept of canonical projection 
and invariant subspace for an inner product space on real numbers. 
These results can be generalized to more general settings to be applied to 
the quantum estimation problem. 

Consider an $n$-dimensional vector space $V$ on $\bbr$ and 
an inner product on it $\vin{\cdot}{\cdot}:V^2\to\bbr$. 
Let $w_1,w_2,\dots,w_k\in V$ ($k<n$) be a set of linearly independent vectors 
and define the subspace $W=\mathrm{span}_\bbr\{ w_1,w_2,\dots,w_k\}\subset V$ 
spanned by these vectors. We define a real positive semi-definite matrix 
$G=[\vin{w_i}{w_j}]_{i,j\in\{1,2,\dots,k\}}$ and its inverse by $G^{-1}=[g^{ij}]_{i,j}$. 
A set of vectors $\{w^i=\sum_j g^{ji}w_j\}_{i=1}^k$ forms the dual basis of $\{w_i\}_{i=1}^k$. 

Given a vector $v\in V$, the canonical projection of $v$ onto the subspace $W$ is 
a map $\pi_W:\ V\to W$ such that
\be
\pi_W(v)=\sum_{i=1}^k \vin{v}{w^i}w_i \in W. 
\ee
This canonical projection is unique and it preserves the inner product as 
\be
\vin{\pi_W(v)}{w}=\vin{v}{w}\ \forall w\in W,\ \forall v\in V. 
\ee
Consider any element $v\in V$ and the condition $v\in W$ is equivalently expressed as follows. 
\begin{align*} 
v\in W&\Leftrightarrow \pi_W(v)=v \\
&\Leftrightarrow \forall v'\in W^\bot\ \vin{v'}{v}=0\\
&\Leftrightarrow \forall v'\in V\ \vin{v'-\pi_W(v')}{v}=0\\
&\Leftrightarrow \forall (v^1,v^2,\dots,v^k)\in V^k,\\
&\quad \forall i(\pi_W(v^i)=w^i\Rightarrow \vin{v^i-w^i}{v}=0),
\end{align*}
where $W^\bot=\{v\in V\,|\, \vin{v}{w}=0\forall w\in W\}$ is the (unique) orthogonal complement of $W$. 

Next, consider a linear map $A$ from $V$ to itself. 
The subspace is said invariant under the map $A$, 
if the image of $W$ is a subset of $W$, i.e., $A(w)\in W$ 
holds for $\forall w\in W$. 
Using the above equivalence, this can be written as 
\begin{align}\label{appDinv}
&W\mbox{ is an}\mbox{ invariant subspace under $A$}\\\nonumber
&\Leftrightarrow \forall i\ A(w_i)\in W\\\nonumber
&\Leftrightarrow \forall i\forall v'\in W^\bot\  \vin{v'}{A(w_i)}=0\\\nonumber
&\Leftrightarrow \forall (v^1,v^2,\dots,v^k)\in V^k,\\\nonumber
&\quad \forall i,j(\pi_W(v^i)=w^i\Rightarrow \vin{v^i-w^i}{A(w_j)}=0).
\end{align}

\subsection{Characterization of D-invariant model} \label{sec:appc-1}
\begin{lemma}\label{lem4}
Let $L_\theta^i$ ($\tilde{L}_\theta^i$) be the SLD (RLD) dual operator 
and $G_\theta$ ($\tilde{G}_\theta$) be the SLD (RLD) Fisher information matrix, respectively. 
Define a $k\times k$ hermite matrix 
by $Z_\theta=[\rldin{L_\theta^i}{L_\theta^j}]_{i,j\in\{1,\dots,k\}}$,
then, the following conditions are equivalent. 
\begin{enumerate}
\item $\cM$ is D-invariant at $\theta$.
\item $\cD{\theta}(L_\theta^i)=\sum_j (\Im Z_\theta)^{ji}L_{\theta,j}$, $\forall i\in\{1,2,\dots,k\}$. 
\item $Z_\theta=\tilde{G}_\theta^{-1}$ 
\item $L_\theta^i=\tilde{L}_\theta^i$, $\forall i\in\{1,2,\dots,k\}$. 
\item $\forall X^i\in\lofhh$, $X^i-L_\theta^i\bot T_\theta(\cM)$ with respect to $\sldin{\cdot}{\cdot}$ 
$\Rightarrow$ $X^i-L_\theta^i\bot T_\theta(\cM)$ with respect to $\rldin{\cdot}{\cdot}$. 
\end{enumerate}
\end{lemma}
In the above lemma, $\sldin{X}{Y}:=\Re\tr{\rho_\theta YX^*}$ and $\rldin{X}{Y}:=\tr{\rho_\theta YX^*}$ denote the SLD and RLD inner products, respectively. 
\begin{proof}
We prove this lemma by the chain; 
$1\Rightarrow 2\Rightarrow3\Rightarrow4\Rightarrow5\Rightarrow1$. 
Suppose that a given model is D-invariant, 
this is equivalent to say that the action of 
the commutation operator on the SLD dual operators is 
expressed as
\be
\cD{\theta}(L_\theta^i)=\sum_{j}c^{ji}L_{\theta,j},
\ee
with some real coefficients $c^{ij}$. These coefficients are 
expressed as 
\be
c^{ji}=\sldin{L_\theta^i}{\cD{\theta}(L_\theta^j)},
\ee
which directly from the orthogonality condition \eqref{orthcond}.
Using the relation \eqref{app1}, the right hand side is also expressed as 
$\Im z_\theta^{ij}$, and if the model is D-invariant at $\theta$, 
\be \label{app3}
\cD{\theta}(L_\theta^i)=\sum_j (\Im Z_\theta)^{ji}L_{\theta,j}. 
\ee
Hence we show $1\Rightarrow 2$. 

Next, if the condition \eqref{app3} holds, the SLD inner product between $\cD{\theta}(L_\theta^i)$ 
and the RLD dual operator $\tilde{L}_\theta^j$ gives 
\be
\sldin{\tilde{L}_\theta^i}{\cD{\theta}(L_\theta^j)}=\sum_k c^{kj} \sldin{\tilde{L}_\theta^i}{L_\theta^k}=c^{ji}. 
\ee
The left hand side is also calculated from Eq.~\eqref{app2} as $-\I(\tilde{g}_\theta^{ij}-g_\theta^{ij})$. 
Therefore, we show 
\be \label{app4}
c^{ji}=\Im z_\theta^{ij}=-\I(\tilde{g}_\theta^{ij}-g_\theta^{ij})\ 
\Leftrightarrow\ \tilde{G}_\theta^{-1}=Z_\theta, 
\ee
holds, if the condition \eqref{app3} holds, that is, $2\Rightarrow 3$. 

Consider an arbitrary linear operator $X$ and assume the condition \eqref{app4}. 
In this case, the RLD inner product between $X$ and $\tilde{L}_\theta^i$ is calculated as
\begin{align*}
\rldin{X}{\tilde{L}_\theta^i}&=\sum_k \tilde{g}_\theta^{kj}\rldin{X}{\tilde{L}_{\theta,k}}\\
&=\sum_k z_\theta^{kj}\rldin{X}{\tilde{L}_{\theta,k}}\\
&=\sum_k z_\theta^{kj}\sldin{X}{L_{\theta,k}}\\
&=\sum_k (g_\theta^{kj}\sldin{X}{L_{\theta,k}}+\I\Im z^{kj} \sldin{X}{L_{\theta,k}})\\
&=\sldin{X}{L_{\theta}^j}+\I \sldin{X}{\cD{\theta}(L_{\theta}^j)}\\
&=\rldin{X}{L_{\theta}^j}, 
\end{align*}
where Eq.~\eqref{app4}, and the several equations presented in Sec.~\ref{sec2-3-2} are used. 
Since $X\in\lofh$ is arbitrary, it implies 
\be\label{app5}
L_\theta^i=\tilde{L}_\theta^i\mbox{ folds for all }i\in\{1,2,\dots,k\}.
\ee 
Therefore, we show $3\Rightarrow4$. 

Next, let us assume the condition \eqref{app5} and we show that this implies the condition $5$, 
that is, $\forall X^i\in\lofhh$, $X^i-L_\theta^i\bot T_\theta(\cM)$ with respect to $\sldin{\cdot}{\cdot}$ 
$\Rightarrow$ $X^i-L_\theta^i\bot T_\theta(\cM)$ with respect to $\rldin{\cdot}{\cdot}$. 
This is because 
\begin{align*}
X^i-&L_\theta^i\bot T_\theta(\cM) \mathrm{w.r.t.} \sldin{\cdot}{\cdot}\\
&\Leftrightarrow \forall j\sldin{X^i-L_\theta^i}{L_{\theta,j}}\\
&\Leftrightarrow \forall j\rldin{X^i-L_\theta^i}{\tilde{L}_{\theta,j}}\\
&\Rightarrow \forall j\rldin{X^i-L_\theta^i}{\tilde{L}_{\theta}^j} \\
&\Leftrightarrow \forall j\rldin{X^i-L_\theta^i}{L_{\theta}^j}\\
&\Leftrightarrow X^i-L_\theta^i\bot T_\theta(\cM) \mathrm{w.r.t.} \rldin{\cdot}{\cdot}. 
\end{align*}

The remaining to be shown is that the condition, 
\begin{multline}\label{app-6}
\forall X^i\in\lofhh,\ \forall i(X^i-L_\theta^i\bot T_\theta(\cM)\mbox{ w.r.t. }\sldin{\cdot}{\cdot}\\
\Rightarrow\ X^i-L_\theta^i\bot T_\theta(\cM)\mbox{ w.r.t. }\rldin{\cdot}{\cdot}), 
\end{multline}
implies the D-invariance of the model. 
Consider a set of hermite operators $\vec{X}=(X^1,X^2,\dots,X^k)$ 
and suppose the above condition \eqref{app-6}. 
Since $X^i-L_\theta^i\bot T_\theta(\cM)\mbox{ w.r.t. }\sldin{\cdot}{\cdot}$ 
is equivalent to $\pi_{T_\theta}(X^i)=L_\theta^i$ with $\pi_{T_\theta}$ 
the canonical projection on $T_\theta(\cM)$, we can rewrite it as 
\be\label{app-7}
\forall X^i\in\lofhh,\ \forall i,j (\sldin{X^i-L_\theta^i}{L_{\theta,j}}\ 
\Rightarrow\ \rldin{X^i-L_\theta^i}{L_{\theta,j}}). 
\ee
The use of Eq.~\eqref{app-0} leads  
$\sldin{X^i-L_\theta^i}{L_{\theta,j}}\Rightarrow \rldin{X^i-L_\theta^i}{L_{\theta,j}}$ 
$\Rightarrow$ $\sldin{X^i-L_\theta^i}{L_{\theta,j}}\Rightarrow \sldin{X^i-L_\theta^i}{\cD{\theta}(L_{\theta,j})}$. 
Then, the equivalent condition \eqref{appDinv} can be applied to conclude that 
the subspace $T_\theta(\cM)$ is invariant under the action of linear operator $\cD{\theta}$, 
that is, the model is D-invariant.  $\square$ 
\end{proof}

\subsection{Proof for Theorem \ref{thmDinv}} \label{sec:appc-2}
\noindent
i) Proof for the RLD CR bound case:\\
The sufficiency (D-invariant model $\Rightarrow C_\theta^H[W]= C_\theta^R[W]$ for all $W>0$.) 
follows from Proposition \ref{propDinv}. If the Holevo bound is identical to the RLD CR bound 
for all weight matrices, then all the RLD dual operators must be hermite, 
i.e., $(\tilde{L}_\theta^i)^* =\tilde{L}_\theta^i$ for all $i=1,2,\dots,k$. 
To see this, we note that $\vec{X}=(X^1,X^2,\dots,X^k)\in \cX_\theta$ 
implies $\vec{X}\in\tilde{\cX}_\theta:=
\{\vec{X}\,|\, \forall i\,X^i\in\lofh,\tr{\rho_\theta X^i}=0,\ \forall i,j\,\tr{\del_i\rho_\theta X^j}=\delta^j_{\,i} \}$. 
By rewriting $\tr{\del_i\rho_\theta X^j}=\delta^j_{\,i}\Leftrightarrow 
\rldin{X^i-\tilde{L}_\theta^i}{\tilde{L}_{\theta,j}}=0$, we see that 
$X^i-\tilde{L}_\theta^i$ ($i=1,2,\dots,k$) are orthogonal to the complex span of the RLD operators, 
$\tilde{T}_\theta(\cM):=\mathrm{span}_\bbc\{\tilde{L}_{\theta,i} \}$.
It is straightforward to show that 
$Z_\theta[\vec{X}]=Z_\theta[\vec{X}-\vec{\tilde{L}}_\theta]+Z_\theta[\vec{\tilde{L}}_\theta]$ 
holds for $\vec{X}\in\tilde{\cX}_\theta$, 
where $\vec{\tilde{L}}_\theta=(\tilde{L}_\theta^1,\tilde{L}_\theta^2,\dots,\tilde{L}_\theta^k)$ 
is the collection of the RLD dual operators. 
Since the matrix $Z_\theta[\vec{X}-\vec{\tilde{L}}_\theta]$ is positive semi-definite, and hence, we obtain 
\be
\vec{X}\in\tilde{\cX}_\theta\Rightarrow
Z_\theta[\vec{X}]\ge \tilde{G}_\theta^{-1},
\ee
and the equality if and only if $\forall i\,X^i=\tilde{L}_\theta^i$ holds. 
Thus, all $\tilde{L}_\theta^i$ need to be in the set $\cX_\theta$. 

Next, we use the relation between the SLD and RLD operators, 
$(I+\I\cD\theta)(\tilde{L}_{\theta,i})=L_{\theta,i}$ (see Eq.~\eqref{app-0}), to get 
\be
(I+\I\cD\theta)(\tilde{L}_{\theta}^i)=\sum_j \tilde{g}^{ji}L_{\theta,j}.  
\ee 
Then, the conditions $(\tilde{L}_\theta^i)^* =\tilde{L}_\theta^i$ for all $i$ imply 
\begin{align}
\tilde{L}_{\theta}^i&=\sum_j \Re\{\tilde{g}^{ji}\}L_{\theta,j},\\
\cD\theta(\tilde{L}_{\theta}^i)&=\sum_j \Im\{\tilde{g}^{ji}\}L_{\theta,j}.
\end{align}
Since $\tilde{G}_\theta$ is positive definite, so as $\Re\tilde{G}_\theta$. 
With these relations, the action of the commutation operator on the SLD operators 
are calculated as
\be
\cD\theta (L_{\theta,i})=\sum_{j,k} [(\Re\{\tilde{G}^{-1}\})^{-1}]_{ji} \Im\{\tilde{g}^{kj}\}L_{\theta,k}.
\ee
This proves that $\cD\theta (L_{\theta,i})\in T_\theta(M)=\mathrm{span}_\bbr\{L_{\theta,i} \}$. 
Therefore, if $C_\theta^H[W]= C_\theta^R[W]$ for all $W\in\cW$, then the model is D-invariant. 
~\\

\noindent
ii) Proof for the D-invariant bound:\\
This equivalence is a direct consequence of the proposition \ref{propDinv} 
($Z_\theta=\tilde{G}_\theta^{-1}$) and the property of the canonical projection given in Appendix \ref{sec:app-3}. 
First, let us note 
\be\label{Dequiv1}
\forall W\in\cW\ C_\theta^H[W]=C_\theta^Z[W]
\Leftrightarrow \forall \vec{X}\in\cX_\theta\ Z_\theta[\vec{X}]\ge Z_\theta. 
\ee
This is because $\vec{L}_\theta$ is an element of the set $\cX_\theta$ 
and $Z_\theta[\vec{L}_\theta]= Z_\theta$. 
If the condition \ref{lem4}-5 holds, it is easy to show 
$\vec{X}\in\cX_\theta\Rightarrow \forall i,j\sldin{X^i-L_\theta^i}{L_{\theta,j}}=0
\Rightarrow \forall i,j\rldin{X^i-L_\theta^i}{L_{\theta,j}}=0$. Therefore, 
we obtain the matrix inequality: 
\begin{align}\nonumber
Z_\theta[\vec{X}]&=[\rldin{X^i}{X^j}]_{i,j}\\ \nonumber
&=[\rldin{X^i-L_\theta^i}{X^j-L_\theta^j}]_{i,j}+[\rldin{L_\theta^i}{L_\theta^j}]_{i,j} \\ 
&\ge [\rldin{L_\theta^i}{L_\theta^j}]_{i,j}=Z_\theta, 
\end{align}
holds for all $\vec{X}\in\cX_\theta$ 
because of the semi-definite positivity of the matrix $[\rldin{X^i-L_\theta^i}{X^j-L_\theta^j}]_{i,j}$. 

Conversely, let us assume $ \forall \vec{X}\in\cX_\theta\ Z_\theta[\vec{X}]\ge Z_\theta$ is true. 
This is possible if and only if $\forall \vec{X}\in\cX_\theta\Rightarrow\forall i,j\ \rldin{X^i-L_\theta^i}{L_{\theta,j}}=0$, 
otherwise we can always find some $\vec{X}\in\cX_\theta$ such that 
the matrix inequality $Z_\theta[\vec{X}]\ge Z_\theta$ is violated. 
Thus, Lemma \ref{lem4}-5 and the equivalence \eqref{Dequiv1} prove this theorem. 
$\square$


\begin{thebibliography}{99}
\bibitem{helstrom}
C.~W.~Helstrom, 
{\it Quantum Detection and Estimation Theory}, (Academic Press, New York, 1976). 

\bibitem{nagaoka89-2}
H.~Nagaoka, 
in {\it Proc.~12th Symp.~ on Inform.~Theory and its Appl.}, 577 (1989). 
Reprinted in the book \cite{hayashi}. 

\bibitem{HM98}
M.~Hayashi and K.~Matsumoto, 
in {\it Surikaiseki Kenkyusho Kokyuroku}, {\bf 1055}, 96 (1998). 
English translation is available in the book \cite{hayashi}.  

\bibitem{GM00}
R.~D.~Gill and S.~Massar, 
Phys.~Rev.~A \textbf{61}, 042312 (2000). 

\bibitem{BNG00}
O.~E.~Barndorff-Nielsen and R.D.~Gill, 
J.~Phys.~A: Math.~Gen.~{\bf 33}, 4481 (2000). 

\bibitem{hayashi03} 
M.~Hayashi, in {\it Selected Papers on Probability and Statistics American Mathematical Society Translations Series 2}, 
Vol. {\bf 277}, 95-123, (Amer. Math. Soc. 2009). (It was originally published in Japanese in Bulletin of Mathematical
Society of Japan, Sugaku, Vol. {\bf 55}, No. 4, 368-391 (2003).) 

\bibitem{hayashi}
M.~Hayashi ed. {\it Asymptotic Theory of Quantum Statistical Inference: Selected Papers}, (World Scientific, 2005). 

\bibitem{bbmr04}
E.~Bagan, M.~Baig, R.~Mu\~{n}oz-Tapia, and A.~Rodriguez, 
Phys.~Rev.~A \textbf{69}, 010304 (2004).

\bibitem{bbgmm06}
E.~Bagan, M.~A.~Ballester, R.~D.~Gill, A.~Monras, and R.~Mu\~{n}oz-Tapia, 
Phys.~Rev.~A \textbf{73}, 032301 (2006).

\bibitem{HM08}
M.~Hayashi and K.~Matsumoto, J.~Math.~Phys. {\bf 49}, 102101 (2008). 

\bibitem{GK06} 
M.~Gu\c{t}\u{a} and J.~Kahn, Phys.~Rev.~A {\bf 73}, 052108 (2006).

\bibitem{KG09} 
J.~Kahn and M.~Gu\c{t}\u{a}, Comm.~Math.~Phys. {\bf 289}, 597 (2009).

\bibitem{YFG13} 
K.~Yamagata, A.~Fujiwara, and R.~D.~Gill, Ann. Stat. {\bf 41}, 2197 (2013).

\bibitem{GG13}
R.~D.~Gill and M.~I.~Gu\c{t}\u{a}, 
IMS Collections, 
From Probability to Statistics and Back: High-Dimensional Models and Processes, 
Vol. {\bf 9} 105 (2013). 

\bibitem{holevo}
A.~S.~Holevo, 
{\it Probabilistic and Statistical Aspects of Quantum Theory}, (Edizioni della Normale, Pisa, 2nd ed, 2011). 

\bibitem{comment0}
An important consequence of regularity conditions is to 
avoid any singular behavior for the quantum Fisher information. 
See Refs.~\cite{holevo,YFG13} for more rigorous mathematical details

\bibitem{yl73}
H.~Yuen and M.~Lax, 
IEEE Trans. on Information Theory, IT{\bf 19}, 740 (1973).

\bibitem{comment1}
We remark that it is also possible to analyze POVMs whose measurement outcomes take values in $\Theta$. 
See, for example, Ch.~6 of Holevo \cite{holevo}. 

\bibitem{Vaart}A.~W.~van der Vaart, 
\textit{Asymptotic Statistics} (Cambridge University Press, 1998).

\bibitem{ANbook}
S.~Amari and H.~Nagaoka, 
{\it Methods of Informatioon Geometry}, 
Translations of Mathematical Monograph, Vol.{\bf 191} 
(AM Sand Oxford University Press, 2000).

\bibitem{nagaoka89}
H.~Nagaoka, 
IEICE Technical Report, \textbf{IT 89-42}, 9 (1989). 
Reprinted in the book \cite{hayashi}. 

\bibitem{yang}
T.~Y.~Young, 
Information Sciences, {\bf 9}, 25 (1975). 

\bibitem{nagaoka87} 
H.~Nagaoka, 
in {\it Proc.~10th Symp.~ on Inform.~Theory and its Appl.}, 241 (1987). 
English translation is available in the book \cite{hayashi}. 

\bibitem{bc94}
S.~L.~Braunstein and C.~M.~Caves, 
Phys.~Rev.~Lett. {\bf 72}, 3439 (1994).

\bibitem{FN99}
A.~Fujiwara and H.~Nagaoka, 
J.~Math.~Phys. {\bf 40}, 4227 (1999). 

\bibitem{nagaokaseminar}
H.~Nagaoka, a series of seminars at the University of Electro-Communications (2013). 

\bibitem{watanabeD}
Y.~Watanabe, 
Ph.D. thesis, the University of Tokyo, (2012). 

\bibitem{matsumoto02}
K.~Matsumoto, J.~Phys.~A: Math.~Gen.~{\bf 35}, 3111 (2002). 

\bibitem{fn95}
A.~Fujiwara and H.~Nagaoka, 
Phys.~Lett. A \textbf{201}, 119 (1995). 

\bibitem{nagaokacomment}
This equivalence, the third line of Eq.~\eqref{thmeq11}, was suggested by H.~Nagaoka (private communication, 2015).  

\bibitem{cdbw14}
P.~J.~D.~Crowley, A.~Datta, M.~Barbieri, and I.~A.~Walmsley, 
Phys.~Rev.~A \textbf{89}, 023845 (2014). 

\bibitem{vdgjkkdbw14}
M.~D.~Vidrighin, G.~Donati, M.~G.~Genoni, X.-M.~Jin, W.~S.~Kolthammer, M.~S.~Kim, A.~Datta, M.~Barbieri, and I.~A.~Walmsley, 
Nat.~Comm.~{\bf 5}, 3532 (2014).

\bibitem{nagaoka91}
H.~Nagaoka, 
Trans.~Jap.~Soc.~Indust.~Appl.~Math., {\bf 1}, 43 (1991).
English translation is available in the book \cite{hayashi}.  

\end{thebibliography}
\end{document}